\documentclass{article}
\usepackage{amsmath}
\usepackage{amsfonts}
\IfFileExists{pdfsync.sty}{\usepackage{pdfsync}}{}
\usepackage{hyperref}

\usepackage{amssymb}
\newenvironment{proof}{{\bf \LANGUE{Proof}{Démonstration}}:}{\hfill $\Box$ }
\providecommand\LANGUE[2]{#1}
\renewcommand{\restriction}{\mathord{\upharpoonright}}
\providecommand\spnewtheorem[4]{\newtheorem{#1}{#2}}
\spnewtheorem{definition}{Definition}{\bfseries}{\rmfamily}
\spnewtheorem{theorem}{\LANGUE{Theorem}{Théorème}}{\bfseries}{\rmfamily}
\spnewtheorem{remark}{\LANGUE{Remark}{Remarque}}{\bfseries}{\rmfamily}
\spnewtheorem{corollary}{\LANGUE{Corollary}{Corollaire}}{\bfseries}{\rmfamily}
\spnewtheorem{lemma}{\LANGUE{Lemma}{Lemme}}{\bfseries}{\rmfamily}
\spnewtheorem{example}{\LANGUE{Example}{Exemple}}{\bfseries}{\rmfamily}

\newcommand\myaddress{ 
LIX,  Ecole Polytechnique, Institut Polytechnique de Paris, France \\
}

\newcommand\hisaddress{University Paris-Est Créteil Val de Marne, LACL, France 
}

\usepackage[utf8]{inputenc}
\usepackage[T1]{fontenc}

\begin{document}

\title{Solvable Initial Value Problems Ruled by Discontinuous Ordinary Differential Equations}
\author{
 Olivier Bournez$^1$\thanks{Partially supported by ANR Project $\partial$IFFERENCE}, Riccardo Gozzi$^{1,2}$\footnotemark[1]
}
\date{
$^1$\myaddress 
$^2$\hisaddress
}

\maketitle

\newcommand{\interior}[1]{%
  {\kern0pt#1}^{\mathrm{o}}%
}
\newcommand\PRIMOF[1]{\tu f_{[#1]}}

\begin{abstract}
We study initial value problems having dynamics ruled by discontinuous ordinary differential equations with the property of possessing a unique solution. We identify a precise class of such systems that we call \emph{solvable intitial value problems} and we prove that for this class of problems the unique solution can always be obtained analytically via transfinite recursion. We present several examples including a nontrivial one whose solution yields, at an integer time, a real encoding of the halting set for Turing machines; therefore showcasing that the behavior of solvable systems is related to ordinal Turing computations. 
\end{abstract}

\section{Introduction}

Recent research has demonstrated the capability of programming with Ordinary Differential Equations (ODEs): a polynomial ODE can be constructed to simulate the evolution of any discrete model, such as Turing machines. This development underscores a fascinating equivalence between discrete computational models—which gain general relevance through the Church-Turing thesis—and  continuous-time analog computational models, namely models based on polynomial ODEs, known to be equivalent to the General Purpose Analog Computer (GPAC) of Claue Shannon \cite{Sha41}. Since the initial establishment of this equivalence, numerous studies have explored its implications, examining the computational attributes and capacities of ODE systems. Notable investigations include: defining computability and complexity classes via ODEs \cite{EmmanuelThese,TheseDaniel,TheseAmaury,TheseRiccardo}; establishing the existence of a universal ODE in the sense of Rubel \cite{LMCS2018}; demonstrating the strong Turing completeness of biochemical reactions \cite{CMSB17vantardise}; and addressing various completeness aspects of reachability problems, such as PTIME-completeness of bounded reachability \cite{JournalACM2017}.


These findings have primarily concentrated on polynomial ordinary differential equations (ODEs) as a representative paradigm for analog computation. A compelling inquiry arises: could it be possible to develop models more powerful than Turing machines by employing more general classes of ODEs? This question holds significant interest both from the perspective of computability (inquiring whether we can compute more) and complexity (inquiring whether we can compute faster).

In this paper, we demonstrate that it is feasible to resolve undecidable problems using discontinuous ordinary differential equations (ODEs) by simulating transfinite computations with a specific class of ODEs. We establish a mathematically robust framework for ordinary differential equations characterized by a unique and unambiguous solution. We define this robust category of ODEs based on what we term "solvable functions" and prove that these support an analytical, transfinite method leading to the solution. We elaborate on this transfinite method, noting that the maximum number of necessary transfinite steps is countable. Our primary findings are succinctly presented in the main theorem in Section \ref{sec:main}. Additionally, we illustrate the potential to construct natural instances of these solvable initial value problems (IVPs) with a unique solution. A detailed example is provided where the solution, at a specific integer time, encodes the Turing machines' halting set. Following the methodology in \cite{TAMC06}, we argue that this class of IVPs exhibits Super-Turing computational capabilities, enabling the simulation of oracle machines that decide Turing jumps.

\paragraph{Remark} 
The current article is an extended version of our article presented at STACS'2024. Compared to \cite{StacsBournezGozzi2024}, we provide more detailed statements, as well as all the proofs, and we provide furthermore many more examples of solvable IVPs. 

\subsection{Structure of the paper}

The paper is organized as follows: the last subsection of the introduction describes some historical accounts of the related areas and topics; Section \ref{sec:prel} presents the necessary knowledge for understanding the article, with a subsection for the notation, one for basic concepts in analysis, one for ordinary differential equations and one for functions of class Baire one; Section \ref{sec:solvable} describes the conditions imposed on the right-hand terms (i.e.\ introducing solvable functions) and proves some auxiliary theorems; Section \ref{sec:search} presents the search method, i.e.\ the main tool to be applied in the domain to obtain analytically the solution; Section \ref{sec:main} describes in detail the transfinite method converging to the solution; Section \ref{sec:examples} illustrates some useful examples of discontinuous IVPs of the type considered and Section \ref{sec:jumps} discusses the relation between solvable systems, Turing jumps and hypercomputation. Section \ref{sec:conclusion} is then a conclusion.

\subsection{Historical accounts and related work}

\subsubsection{Denjoy's totalization method for integration.} 
Solving ordinary differential equations (ODEs) can be seen as a more general process than antidifferentiation; specifically, integrating is a simpler form of ODE solving where the derivative is explicitly known. This ties into a long-standing mathematical question about whether it's possible to reconstruct a function $f$
 from its derivative $f'$
  under the most general conditions. Unfortunately, traditional methods like the Riemann and Lebesgue integrals fall short as they require certain conditions on the derivative to be effective. There are well-documented cases of functions that are differentiable across a closed interval, yet their derivatives are not Lebesgue integrable.

Historically, Denjoy was the pioneer in addressing these complex cases by developing an integral concept that both encompassed and extended the previous methods, allowing the reconstruction of $f$ from $f'$ 
using a transfinite process. This method, which Denjoy termed the "totalization" method, was wholly analytical, employing transfinite iterations involving limit taking and repeated Lebesgue integrations. This method is capable of deriving 
$f$ from $f'$  through at most a countable number of transfinite steps.

Our proposed method for solving initial value problems (IVPs) draws inspiration from Denjoy's approach. The class of solvable IVPs we focus on shares a crucial characteristic with Denjoy’s integrable derivatives: the solution to these systems can always be achieved through a transfinite computational process. Such IVPs bridge the gap between transfinite computations in digital models and analog models utilizing solvable dynamics, underscoring a deep connection across different computational frameworks.

\subsubsection{ODEs as analog model of computation.} 
The concept of utilizing ordinary differential equations (ODEs) as a computational model was initially explored by Claude Shannon in his work on the General Purpose Analog Computer (GPAC). Shannon's theoretical aim was to establish a universal model of computation that could replicate the behavior of integrator-based devices. He noted that all functions generated by these devices are differential-algebraic, thus laying the foundational stone for later recognizing systems of polynomial ODEs as a valid analog computational model.

The progression from algebraic differential equations to polynomial ODEs, however, is complex and technical:
In a pivotal study (\cite{Pou74}), it was revealed that Shannon's original model lacked completeness and formal structure, necessitating modifications. Later, the researchers in \cite{GC03} addressed these shortcomings by imposing restrictions on the connections within the circuits of the model. These adjustments not only refined the model but also led to the intriguing outcome of narrowing down the class of functions permissible within the GPAC framework. It was conclusively shown that this refined class aligns precisely with the solutions of polynomial ODEs.

\subsubsection{What is known about polynomial ODEs: lower bounds.} The GPAC model was later proved to be equivalent to computable analysis, capable of generating all computable functions over the reals (\cite{TAMC06}). This equivalence was pivotal for establishing a practical link between analog computation models like the GPAC and discrete computation models such as Turing machines. Specifically, the foundational work in \cite{TAMC06} showcased how Turing machine computations could be simulated using initial value problems (IVPs) crafted from polynomial ODEs. These findings paved the way for further exploration into the complexities inherent in the GPAC model. A key discovery was that the length of the solutions to these ODEs serves as an appropriate measure of complexity (\cite{JournalACM2017}).

This relationship between solution length and complexity has been effectively leveraged to categorize various complexity classes within the GPAC model, each corresponding naturally to discrete time-complexity classes such as $\operatorname{FP}$ or $\operatorname{FEXP}$ (\cite{JournalACM2017}, \cite{GozziGraca2022}). The introduction of robust conditions for the dynamical systems used to emulate Turing machines filled a crucial gap, which was also sufficient to establish a correspondence with the polynomial-space-complexity class $\operatorname{FPSPACE}$, as detailed in \cite{BGDPRiccardo2022}. Collectively, these results robustly validate polynomial ODEs as a viable paradigm for analog computation and confirm their Turing completeness.

\subsubsection{What is known about ODEs: upper bounds.}
On another front, extensive research has explored the computability and complexity properties of solving ordinary differential equations (ODEs) more generally. This line of inquiry takes a fundamentally different approach from our previous discussions. Rather than adapting discrete computations to a continuous framework, this research seeks to identify which classes of initial value problems (IVPs) involving ODEs can be solved algorithmically. Many problems for ODEs turn out to be undecidable. For instance, it has been shown that determining the boundedness of the definition domain remains undecidable for polynomial ODEs (\cite{gracca2008boundedness}).

For the class of polynomial-time computable, Lipschitz continuous ODEs, the solutions have been established as computable since \cite{Ko83}. A detailed analysis of the complexity involved was carried out in \cite{Kaw10}, revealing that solving such problems is PSPACE-complete. The requirement for the Lipschitz condition suggests that the uniqueness of a solution to a given IVP might be necessary to compute it, although it is not a sufficient criterion. There are cases where an IVP has a unique solution and computable input data, yet the solution itself remains uncomputable, such as the scenario described in \cite{PR81}.

Further exploration of the divide between necessary and sufficient conditions for computability was conducted in \cite{collins2009effective}. The findings indicated that solutions of continuous ODEs with unique solutions are necessarily computable. The algorithm developed in \cite{collins2009effective}, whimsically named the "Ten Thousand Monkeys algorithm," employs a search strategy across the entire solution space by enumerating all finite sequences of open rational boxes.

\section{Preliminaries}
\label{sec:prel}

\subsection{Notation}

We introduce the notation and the main definitions used throughout this work. The standard symbols $\mathbb{N}$, $\mathbb{N}_0$, $\mathbb{R}$ and $\mathbb{Q}$ stand for the set of natural, natural including zero, real and rational numbers respectively while $\mathbb{R}^+$ and $\mathbb{Q}^+$ represent the positive reals and the positive rationals. When making use of the norm operator we always consider euclidean norms if not specified otherwise. We use $\langle n_1,\hdots,n_k \rangle$ to denote a single integer which represents the tuple $( n_1,\hdots,n_k)$ according to some standard computable encoding. We refer to a \emph{compact domain} of a topological space as a nonempty connected compact subset of such space. Given a metric space $X$ we indicate with the notation $d_X$ the distance function in such space and with the notation $B_X(x, \delta)$ the open ball centered in $x \in X$ with radius $\delta >0$. By default we describe as open \emph{rational ball} or as open \emph{rational box} an open ball or box with rational parameters. Precisely, an open rational box $B$ is a set of the form $(a_1, b_1) \times \hdots \times (a_r, b_r) \subset \mathbb{R}^r$ for some $r \in \mathbb{N}$ where $a_i, b_i \in \mathbb{Q}$ for $i = 1, \hdots, r$. We indicate with the notations $\operatorname{diam}(B)$ and $\operatorname{rad}(B)$ its diameter and its radius. Moreover, given a function $f: [a,b] \to \mathbb{R}^r$ for some $a,b \in \mathbb{R}, a<b$ and some $r \in \mathbb{N}$ we indicate with the notation $f': [a,b] \to \mathbb{R}^r$ the derivative of such function, where the derivative on the extremes $a$ and $b$ is defined as the limit on the left and on the right respectively. Given a function $f: X \to Y$ and a set $K \subseteq X$ we indicate with the notation $f \restriction_K$ the restriction of function $f$ to the set $K$, i.e. $f \restriction_K$ is the function from $K$ to $Y$ defined as $f \restriction_K (x)= f(x)$. If $A$ and $B$ are two sets, we refer to the set difference operation using the symbol $A \setminus B$ and we indicate with the notation $A+B$ the Minkowski sum of set $A$ with set $B$. The expression $\operatorname{cl}(A)$ indicates the closure of $A$, $\operatorname{int}(A)$ the interior of $A$ and $A^c$ its complement. The notation $\emptyset$ stands for the empty set, while the notation $\omega_1$ stands for the first uncountable ordinal number. Given a property of a function $f: X \to Y$, we say that this property is satisfied \emph{almost everywhere} if the property is satisfied on $X \setminus D$, where $D$ is a set with Lebesgue measure equal to zero. We use the symbol $\phi_e$ to denote to the eth Turing functional according to some G\"{o}del enumeration. Given a set $X$ of elements and an index set $Y$ we indicate sequences of elements from $X$ with the notation $\{ x_n \}_{n \in Y}$ where $x_n \in X$ for each $n \in Y$. If the index set is the set of natural numbers, we simply write $\{ x_n \}_n$. Instead, if the index set is some ordinal number, we talk about \emph{transfinite sequences}. Given a sequence $\{ x_n \}_{n}$ for some set of elements $X$ we indicate a subsequence of such sequence with the notation $\{ x_{n(u)} \}_u$ where $n: \mathbb{N} \to \mathbb{N}$ is the function determining the elements of the subsequence considered.

\subsection{Basic concepts in analysis and useful theorems}

We introduce some useful definitions and results:

\begin{definition}[Globally Lipschitz function]
Let $[a,b]$ an interval in $\mathbb{R}$, $r \in \mathbb{N}$ and let $f: [a,b] \to \mathbb{R}^r$ be a function. We say that $f$ is \emph{globally Lipschitz} if there exists a constant $M>0$ such that $\left\Vert f(\tilde{t}) - f(t) \right\Vert \leq M \left\vert t - \tilde{t}\right\vert$ for all $a \leq t, \tilde{t} \leq b$. 
\end{definition}

\begin{definition}[Absolutely continuous function]
Let $[a,b]$ an interval in $\mathbb{R}$, $r \in \mathbb{N}$ and let $f: [a,b] \to \mathbb{R}^r$ be a function. We say that $f$ is \emph{absolutely continuous} if for all $\epsilon>0$ there exists a $\delta_\epsilon>0$ such that for all finite sequence of pairwise disjoint intervals $\{ (x_k,y_k) \}_k$ in $[a,b]$ satisfying $\sum_{k} \left\vert y_k - x_k \right\vert < \delta_\epsilon$ we have $\sum_{k} \left\Vert f(y_k) - f(x_k) \right\Vert < \epsilon$.
\end{definition}

The following theorem is a well known result, as in \cite{Hobs21}(Theorem 406)

\begin{theorem}
\label{thm:Lebediff}
Let $[a,b]$ an interval in $\mathbb{R}$, $r \in \mathbb{N}$ and let $f: [a,b] \to \mathbb{R}^r$ be an absolutely continuous function. Then $f$ is differentiable almost everywhere and $f(x)= \int_{a}^{x} f'(t) dt$ for all $x \in [a,b]$. 
\end{theorem}

We now define \emph{uniformly bounded} sequences of functions:

\begin{definition}[Uniformly bounded]
Let $I \subset \mathbb{R}$, $E \subset \mathbb{R}^r$ for some $r \in \mathbb{N}$ and let $\{ g_n \}_n : I \to E$ be a sequence of functions. We say that the sequence is \emph{uniformly bounded} if there exists a constant $K>0$ such that $\left\Vert g_n (t) \right\Vert \leq K$ for all $g_n \in \{ g_n \}_n$ and for all $t \in I$.
\end{definition}

We then define \emph{equicontinuous} sequences of functions:

\begin{definition}[Equicontinuous]
Let $I \subset \mathbb{R}$, $E \subset \mathbb{R}^r$ for some $r \in \mathbb{N}$ and let $\{ g_n \}_n : I \to E$ be a sequence of functions. We say that the sequence is \emph{equicontinuous} if for any $\epsilon >0$ there exists a $\delta_{\epsilon}>0$ such that $\left\Vert g_n (t) - g_n (\tilde{t}) \right\Vert \leq \epsilon $ whenever $\left\vert t - \tilde{t} \right\vert \leq \delta_{\epsilon}$ for all $g_n \in \{ g_n \}_n$ and for all $t, \tilde{t} \in I$.
\end{definition}

For infinite, uniformly bounded, equicontinuous sequence of functions over the reals there exists a famous result due to Ascoli \cite{CL55}:

\begin{theorem}[Ascoli] 
\label{thm:ascoli}
Let $I \subset \mathbb{R}$ be a bounded interval, $E \subset \mathbb{R}^r$ for some $r \in \mathbb{N}$ and let $\{ g_n \}_n : I \to E$ be an infinite, uniformly bounded, equicontinuous sequence of functions. Then the sequence $\{ g_n \}_n$ has a subsequence $\{ g_{n(u)} \}_{u}$ that converges uniformly on $I$.
\end{theorem}

Uniformly convergent sequences satsify the well known uniform limit theorem \cite{MUN}:

\begin{theorem}[Uniform limit theorem] 
\label{thm:uniform}
Let $I \subset \mathbb{R}$ be a bounded interval, $E \subset \mathbb{R}^r$ for some $r \in \mathbb{N}$ and let $\{ g_n \}_n : I \to E$ be a sequence of functions that converges uniformly on $I$ to function $g$. If every function $g_n$ in the sequence is continuous on $I$ then $g$ is continuous on $I$.
\end{theorem}

We are also going to use another very well known result in analysis: 

\begin{theorem}[Differentiable limit theorem]
\label{thm:difflimit}
Let $\{ f_n \}_n$ be a sequence of differentiable functions from the closed interval $[a,b]$ to $\mathbb{R}$ pointwise converging to function $f$. If the sequence of functions $\{ f'_n \}_n$ converges uniformly on $[a, b]$ to a function $r$, then $f$ is differentiable and $f'=r$.
\end{theorem} 

It follows a theorem concerning differentiabilty of limits of converging sequences of functions: 

\begin{theorem}
Let $\{ f_n \}_n$ be sequence of functions from the closed interval $[a,b] \subset \mathbb{R}$ to $\mathbb{R}^r$ for some $r \in \mathbb{N}$ and pointwise converging to function $f$. Let $M >0$ be such that $\left\Vert f_{n}(\tilde{t}) - f_{n}(t)  \right\Vert \leq M \left\vert t - \tilde{t}\right\vert $ for all $n \in \mathbb{N}$, for all $a \leq t, \tilde{t} \leq b$. Then $f$ is differentiable almost everywhere and $f(x)= \int_{a}^{x} f'(t) dt$ for all $\in [a,b]$. 
\label{thm:differsequence}
\end{theorem}

\begin{proof}
First note that $\left\Vert f(\tilde{t}) - f(t)  \right\Vert = \lim_{n \to \infty} \left\Vert f_{n}(\tilde{t}) - f_{n}(t)  \right\Vert \leq M \left\vert t - \tilde{t}\right\vert $ for all $a \leq t, \tilde{t} \leq b$. This means that $f$ is globally Lipschitz with Lipschitz constant $M$. We now show that this implies that $f$ is absolutely continuous. Indeed, for all $\epsilon >0$, for all finite sequence of pairwise disjoint intervals $\{ (x_k,y_k) \}_k$ in $[a,b]$ such that $\sum_{k} \left\vert y_k - x_k \right\vert < \frac{\epsilon}{M}$ we have that $\sum_{k} \left\Vert f(y_k) - f(x_k) \right\Vert \leq M \sum_{k} \left\vert y_k -  x_k \right\vert < \epsilon$. Hence, since $f$ is absolutely continuous it follows from Theorem \ref{thm:Lebediff} that $f$ is differentiable almost everywhere and $f(x)= \int_{a}^{x} f'(t) dt$ for all $x \in [a,b]$. 
\end{proof}
\
\vspace{3mm}

We describe the process of \emph{transfinite recursion} as the process that for each ordinal $\alpha$ associates with $\alpha$ an object that is described in terms of objects already associated with ordinals $\beta < \alpha$. A more detailed explanation can be found in \cite{sacks1990higher}. We use the expression \emph{transfinite recursion up to $\alpha$} if the process associates an object for all ordinals $\beta <\alpha$.

We present the Cantor-Baire stationary principle \cite{EV10}, as expressed by the following theorem: 

\begin{theorem}[Cantor-Baire stationary principle]
\label{thm:statio}
Let $\left\{E_\gamma \right\}_{\gamma <\omega_1}$ be a transfinite sequence of closed subsets of $\mathbb{R}^r$ for some $r \in \mathbb{N}$. Suppose $\left\{E_\gamma \right\}_{\gamma<\omega_1}$ is decreasing; i.e.\, $E_\gamma \subseteq E_\beta$ if $\gamma \geq \beta$. Then there exists $\alpha<\omega_1$ such that $E_\beta=E_{\alpha}$ for all $\beta \geq \alpha$.
\end{theorem}

\subsection{Ordinary Differential Equations and Initial Value Problems}

We consider dynamical systems whose evolution is described by ordinary differential equations. Consider an interval $[a,b] \subset \mathbb{R}$, a compact domain $E \subset \mathbb{R}^r$ for some $r \in \mathbb{N}$, a point $y_0 \in E$ and a function $f: E \to \mathbb{R}^r$ such that the dynamical system:

\begin{equation}
\label{eq:pb}
\begin{cases}
y'(t)= f(y(t))\\
y(a)=y_0
\end{cases}
\end{equation}

has at least one solution $y:[a,b] \to \mathbb{R}^r$ with $y([a,b]) \subset E$. Given $y_0$ and $f$, the problem of obtaining one solution in such setting is called an initial value problem. The condition $y(a)=y_0$ (or, in short, just the point $y_0$) is referred to as the initial condition of the problem and function $f$ is referred to as the right-hand term of the problem. In general there can be multiple valid solutions for the same IVP. Nonetheless, in this work we consider only IVPs whose solution is uniquely defined. Therefore, we refer to function  $y:[a,b] \to E$ satisfying Equation \ref{eq:pb} as the solution of the problem. In this particular setting there exist different ways to obtain the solution analytically when the right-hand term is continuous. Many of these methods, such as building Tonelli sequences, are often introduced for proving Peano's theorem related to the existence of the solution for IVPs with continuous right-hand terms, and are based on the concept of constructing a sequence of continuous functions converging to the solution. Once it is known that the solution is unique, every sequence considered in each of these methods can be shown to converge to the unique solution. An analysis based on these methods can also achieve computability for the solution, where computability has to be intended in the sense of computable analysis, as it has been proved by the authors of \cite{collins2009effective} in the description of their so called \emph{ten thousand monkeys} algorithm. The idea of this algorithm is to exploit the hypothesis of unicity for enclosing the solution into covers of arbitrarily close rational boxes in $E$. 

\begin{definition}[Ten thousand monkeys algorithm]
\label{def:monkeys}
Consider a compact domain $E \subset \mathbb{R}^r$ for some $r \in \mathbb{N}$ and a right-hand term $f: E \to \mathbb{R}^r$ for an ODE of the form of the one in Equation \ref{eq:pb}. Let $y_0$ be a point in $E$. We call the \emph{ten thousand monkeys algorithm for $(f, y_0)$} the following procedure: enumerate all tuples of the form $\left(X_{i, j}, h_i, B_{i, j}, C_{i, j}, Y_{i, j}\right)$ for $i= 0, \ldots, l-1, j=1, \ldots, m_i$, where $l, m_i \in \mathbb{N}, X_{i, j}, B_{i, j}, C_{i, j}$ and $Y_{i, j}$ are open rational boxes in $E$ and $h_i \in \mathbb{Q}^+$. Such a tuple is a run of the algorithm. A run of the algorithm is said to be valid if $y_0 \in \operatorname{int}\left(\bigcup_{j=1}^{m_0} X_{0, j}\right)$, and for all $i=0, \ldots, l-1$ and $j=1, \ldots, m_i$, we have:

\begin{enumerate}
\item $\operatorname{int}\left( f\left(B_{i, j}\right)\right) \subset \operatorname{cl}\left(C_{i, j}\right)$
\item $X_{i, j} \cup Y_{i, j} \subset B_{i, j}$;
\item $X_{i, j}+h_i C_{i, j} \subset Y_{i, j}$;
\item $\bigcup_{j=1}^{m_i} Y_{i, j} \subset \bigcup_{j=1}^{m_{i+1}} X_{i+1, j}$.
\end{enumerate}

The output of the algorithm is the infinite sequence of all valid runs. 

\end{definition}

Concerning this algorithm, the authors of \cite{collins2009effective} proved the following theorem:

\begin{theorem}
\label{thm:monkeys}
Consider an interval $[a,b] \subset \mathbb{R}$, a compact domain $E \subset \mathbb{R}^r$ for some $r \in \mathbb{N}$, a right-hand term $f: E \to \mathbb{R}^r$ and an initial condition $y_0$ for an IVP of the form of Equation \ref{eq:pb}. Consider the ten thousand monkeys algorithm for $(f, y_0)$ and its valid runs of the form $\left(X_{i, j}, h_i, B_{i, j}, C_{i, j}, Y_{i, j}\right)$ for $i= 0, \ldots, l-1, j=1, \ldots, m_i$, where $l, m_i \in \mathbb{N}, X_{i, j}, B_{i, j}, C_{i, j}$ and $Y_{i, j}$ are open rational boxes in $E$ and $h_i \in \mathbb{Q}^+$. For each valid run define $t_0=a$ and $t_i=\sum_{j=0}^{i-1} h_j $ for all $i \in \{ 1, \ldots, l \}$. If $f$ is continuous and is such that there exists a unique solution $y:[a,b] \to E$ of the IVP then we have that: 

\begin{itemize}
\item Any valid run of the algorithm satisfies $y(t) \in B_i$ for all $t_i \leq t \leq t_{i+1}$, for all $i \in \{ 0, \ldots, l-1 \}$.
\item For any $\epsilon > 0$ there is a valid run of the algorithm such that $t_l=b$ and $\operatorname{diam}(B_i)< \epsilon $ for all $i \in \{ 0, \ldots, l-1 \}$
\end{itemize}

where $B_i= \bigcup_ {j=1, \ldots, m_{i}} B_{i, j}$ for all $i= 0, \ldots, l -1$.

\end{theorem}

It is clear from the theorem above that the unique solution can be obtained analytically as a limit produre from the output of the ten thousand monkeys algorithm. Nonetheless, all these methods are correctly functioning as long as the right-hand term of the IVP is continuous. Our first objective is to relax continuity for the right-hand term of the IVP while still being able to obtain the unique solution. Our analysis shows that it is still possible to analytically obtain the solution from a more general class of IVPs with discontinuous right-hand terms when the solution is unique. We show that this is possible via transfinite recursion as long as we impose the correct requirements on the right-hand term of the IVP. 

\begin{remark}[Discontinuous ODEs]

It is important to point out that several mathematical theories exist for discussing discontinuous ODEs: see for example \cite{FilippovBook,aubin2012differential,deimling2011multivalued}. However the concept of solution differs from one theory to the other (there is not a unique theory for discontinuous ODEs) and that the existence of a solution is often a non-trivial problem in all these theories. We are here, and in all the examples, in the case where we know that there exisits a solution and a unique solution, so in a case where there is no ambiguity about the solution concept. Furthermore, all the theories we know consider that equality almost everywhere in \eqref{eq:pb} is sufficient, mostly to be able to define the solution in itegral form by means of Lebesgue integration. We instead require equation \eqref{eq:pb} to hold for all points. 

\end{remark}

\subsection{Functions of class Baire one}

We define the \emph{set of discontinuity points} of a given function:

\begin{definition}[Set of discontinuity points]
\label{def:disco}
Let $f$ be a function $f: X \to Y$ where $X$ and $Y$ are two complete metric spaces. We define the \emph{set of discontinuity points} (of $f$ on $X$) as the the set: 
\begin{equation*}
D_f = \left\{ x \in X : \exists \epsilon > 0 : \forall \delta >0 \; \exists y,z \in \; B_X (x, \delta) : d_Y( f(y), f(z) ) > \epsilon \right\}
\end{equation*}
\end{definition}

We define what it means for a given function to be \emph{of class Baire one}:

\begin{definition} [Baire one]
Let $X$, $Y$ be two separable, complete metric spaces. A function $f: X \to Y$ is \emph{of class Baire one} if it is a pointwise limit of a sequence of continuous functions, i.e.\ if there exists a sequence of continuous functions from $X$ to $Y$, $\{ f_m \}_m$,  such that $\displaystyle\lim_{m \to \infty} f_m (x) = f(x)$ for all $x \in X$.
\end{definition}

Notably, all derivaties are functions of class Baire one by definition. An important property of functions of class Baire one is that the composition of a function of class Baire one with a continuous functions yields a function of class Baire one \cite{ZH07}. We now refresh some well known topological concepts:

\begin{definition}[Nowhere dense set, meager set]
Let $X$ denote a topological space and let $S$ be a subset of $X$.
\begin{itemize}
\item We say that $S$ is \emph{nowhere dense} (in $X$) if its closure has empty interior. 
\item We say that $S$ is \emph{meager} (in $X$) if it is a countable union of nowhere dense sets. 
\end{itemize}
\end{definition}

It follows the definition of what it is means for a topological space to be \emph{Baire}. 

\begin{definition}[Baire space]
Let $X$ be a topological space, we say that this space is \emph{Baire} if countable unions of closed sets with empty interior also have empty interior. 
\end{definition}

We then present one important theorem obtained by Baire, known as the Baire cathegory theorem \cite{BAI99}. 

\begin{theorem}[Baire cathegory theorem]
Let $X$ be a complete metric space. Then we have: 
\begin{enumerate}
\item A meager set in $X$ has empty interior. 
\item The complement of a meager set in $X$ is dense in $X$. 
\item A countable intersection of open dense sets in $X$ is dense in $X$. 
\end{enumerate}
\end{theorem}

It is immediately clear from this theorem combined with the definition of Baire space presented above that any complete metric space is also a Baire space. 

\begin{corollary}
\label{coro:Baire}
Let $X$ be a complete metric space, then $X$ is a Baire space. 
\end{corollary} 

In some texts the above corollary can be found expressed as an alternative, equivalent version of the Baire cathegory theorem we presented above \cite{MUN}[Theorem 48.2]. Another well known consequence of the above theorem is the following corollary \cite{LIT14}[Theorem 2.3]:

\begin{corollary} 
\label{coro:Baire2}
Let $X$ be a nonempty Baire space and let $\{ K_n \}_n$ be a sequence of closed subsets of $X$ such that $X= \displaystyle\cup_{n=1}^{\infty} K_{n}$. Then there exists at least one $n_0 \in \mathbb{N}$ such that $K_{n_0}$ has nonempty interior.
\end{corollary}

Let us now consider the set of discontinuity points $D_{f}$ of a function $f: X \to Y$ where $X$ and $Y$ are two nonempty complete metric spaces. Notice that this set can be expressed as: 

\begin{equation*}
D_f = \displaystyle\bigcup_{n=1}^{\infty} D_{f,n}
\end{equation*}

where each set $D_{f,n}$ is defined as:

\begin{equation*}
 D_{f,n}= \left\{ x \in X : \forall \delta \geq 0 \; \exists y,z \in B_X (x, \delta) : d_Y ( f(y), f(z) ) \geq \frac{1}{n} \right\} 
\end{equation*}

it is easy to see that each of these sets is closed in $X$.  If we require for function $f$ to be of class Baire one, we obtain something more: 

\begin{lemma}

Let $f: X \to Y$ be a function of class Baire one where $X$ and $Y$ are two nonempty, separable complete metric spaces. Then for every $n \in \mathbb{N}$, the set $D_{f,n}$ defined as:

\begin{equation*}
 D_{f,n}= \left\{ x \in X : \forall \delta \geq 0 \; \exists y,z \in B_X (x, \delta) : d_Y ( f(y), f(z) ) \geq \frac{1}{n} \right\} 
\end{equation*}

has empty interior. 

\end{lemma}

\begin{proof}

For each $n \in \mathbb{N}$ the case in which $D_{f,n} = \emptyset$ is trivial. Suppose then $D_{f,n} \neq \emptyset$. We are going to prove the lemma by contradiction. Suppose that there is an open ball $O_n$ such that $O_n \subset D_{f,n}$. Recall now that since $f$ is a function of class Baire one and $D_{f,n}$ is a separable complete metric space we can define a sequence of continuous functions $\{ f_m \}_m$ such that $\displaystyle\lim_{m \to \infty} f_m (x) = f(x)$ for all $x \in D_{f,n}$. Then, we can define a sequence of sets $ \{ C_{n, N} \}_N$, where for every $N \in \mathbb{N}$, each set is defined as: 

\begin{equation*}
C_{n, N}= \left\{ x \in O_n : \forall m,k \geq N \; , \;  d_Y ( f_m (x), f_k (x) ) \leq \frac{1}{4n} \right\}
\end{equation*}

These sets are clearly closed. By pointwise convergence of the sequence $\{ f_m \}_m$ we have, $ O_n = \displaystyle\cup_{N=1}^{\infty} C_{n, N}$. Because every open subset of a Baire space is itself a Baire space, we can make use of Corollary \ref{coro:Baire2} in order to conclude that there exists at least one $N^*$ such that the set $C_{n, N^*}$ contains on open ball $O^*_n$ in $O_n$. That means that we have $ O^*_n \subseteq C_{n, N^*} \subseteq O_n$. 

Notice now that by taking pointwise limits in the definition of $C_{n, N^*}$, for each $x \in O^*_n \subseteq C_{n, N^*}$ we obtain that $d_Y ( f (x), f_{N^*} (x) ) \leq \frac{1}{4n}$. We now show that this implies that $D_{f,n} \cap O^*_n = \emptyset$ which contradicts the hypothesis $O_n \subseteq D_{f,n}$ since $ O^*_n \subseteq O_n$. Indeed, since function $f_{N^*}$ is continuous on $O^*_n$, we know that for each $x \in O^*_n$ there exists a $\delta >0$ such that $\forall y, z \in B_{O^*_n} (x, \delta), \;  d_Y ( f_{N^*}(y), f_{N^*}(z) ) < \frac{1}{4n}$ which in turn implies that $ d_Y ( f(y), f(z) ) < \frac{1}{2n}$. Therefore, by definition of $D_{f,n}$, we have $D_{f,n} \cap O^*_n = \emptyset$ which as mentioned contradicts the hypothesis, proving the lemma.
\end{proof}


The lemma above, combined with Corollary \ref{coro:Baire} immediately yields: 

\begin{lemma}
\label{lemma:Baire}
Let $f: X \to Y$ be a function of class Baire one where $X$ and $Y$ are two nonempty, separable complete metric spaces. Then the set of discontinuity points $D_f$ of function $f$ in $X$ is a meager set which can be expressed as $D_f = \displaystyle\cup_{n=1}^{\infty} I_n$ where, for all $n \in \mathbb{N}$, $I_n$ is a closed nowhere dense set. Moreover, $D_f$ has empty interior. 
\end{lemma}

\section{Solvable functions}
\label{sec:solvable}

We start this section by preparing the right setting for a transfinite classification of the right-hand terms of our systems. This stratification is based upon the degree of discontinuity of the functions considered. We can quantify this precisely by introducing the following definition:

\begin{definition}[Sequence of $f$-removed sets on $E$]
\label{def:removed}

For some $r,m \in \mathbb{N}$, consider a compact domain $E \subset \mathbb{R}^r$ and a function $f: E \to \mathbb{R}^m$. Let $\{ E_\alpha\}_{\alpha}$ be a transfinite sequence of sets and $\{ f_\alpha \}_{\alpha}$ a transfinite sequence of functions such that $f_\alpha = f \restriction_{E_{\alpha}}: E_{\alpha} \to \mathbb{R}^m$ defined as following:

\begin{itemize}
\item $E_0=E$ 
\item For every $\alpha$, $E_{\alpha+1}=  D_{ f_{\alpha}}$
\item For every $\alpha$ limit ordinal, $E_\alpha= \displaystyle\cap_{\beta<\alpha} E_\beta $
\end{itemize}

we call the sequence $\{ E_\alpha\}_{\alpha}$ the \emph{sequence of $f$-removed sets} on $E$.

\end{definition}

We remark that since functions in the sequence $\{ f_\alpha \}_{\alpha}$ above are allowed to be defined over disconnected sets, the notion of continuity in the above definition has to be intended with respect to the induced topology relative to $E_{\alpha}$ as a subset of $\mathbb{R}^r$. Moreover, note that it follows from the definition of the sequence that such sequence is decreasing, meaning that $E_\delta \subseteq E_\gamma$ if $\delta > \gamma$ for every $E_\delta$ and $E_\gamma$ in the sequence. 

\begin{remark}[Similar ranking for measuring discontinuities] \ \\
The definition of this sequence, or of slight variations of the same sequence, has already been considered in the literature. For instance, the author of \cite{Her96} selects a version of this sequence where the closure of the sets is taken at each level and relates such sequence of functions with a bound for the topological complexity of any algorithm that computes them while using only comparisons and continuous arithmetic (and information) operations. Similarly, starting from the same transfinite sequence of functions applied to countably based Kolmogorov spaces, it is shown in \cite{deB14} that a given function is at the $\alpha$ level of the hierarchy if and only if it is realizable through the $\alpha$-jump of a representation. 
\end{remark}

Some interesting properties of this sequence are outlined by the two following lemmas. 

\begin{lemma}
\label{lemma:extremes}
For some $r,m \in \mathbb{N}$, consider a compact domain $E \subset \mathbb{R}^r$ and a function $f: E \to \mathbb{R}^m$. Let $\{ E_\gamma \}_{\gamma}$ be the sequence of $f$-removed sets on $E$. If $x \in E_\alpha$ for some $E_\alpha$ in the sequence then $x$ is an accumulation point of $E_\beta$ for all $\beta < \alpha$. 
\end{lemma}

\begin{proof}
Note that by definition if $x \in E_\alpha$ for some ordinal $\alpha$ that means that for all ordinals $\beta < \alpha$ function $f \restriction_{E_\beta}$ is discontinuous on $x$. Since every function is continuous on every isolated point of its domain of definition, this implies that $x$ is an accumulation point for all $E_\beta$ such that $\beta < \alpha$.
\end{proof}

\begin{lemma}
\label{lemma:somealpha}
For some $r,m \in \mathbb{N}$, consider a compact domain $E \subset \mathbb{R}^r$ and a function $f: E \to \mathbb{R}^m$. Let $\{ E_\gamma \}_{\gamma}$ be the sequence of $f$-removed sets on $E$ and let $E_\alpha = \emptyset$ for some $\alpha$. For every point $x \in E$ there exists an ordinal $\beta<\alpha$ such that $x \in E_\beta \setminus E_{\beta+1}$ and there exists an open ball $B_E (x, \epsilon)$ for some $\epsilon>0$ such that $B_E(x, \epsilon) \cap E_{\beta+1} = \emptyset$.
\end{lemma}

\begin{proof}
Let $x \in E$. If $f$ is continuous on $x$ then trivially $x \in E_0 \setminus E_1$. Instead, let $x$ be a point in the set of discontinuity points of $f$, which means that $x \in E_1$. Suppose that there is no $\beta<\alpha$ such that $x \in E_\beta \setminus E_{\beta+1}$. We prove now that this supposition leads to a contradiction. Indeed, since we know that the sequence of $f$-removed sets on $E$ is a decreasing sequence, and since we know by hypothesis that there exists an ordinal $\alpha$ such that $E_{\alpha}= \emptyset$, it follows that the considered sequence of $f$-removed sets is a strictly decreasing sequence, which in turn means that if there is no $\beta<\alpha$ such that $x \in E_\beta \setminus E_{\beta+1}$ we need to have that $x \notin E_\gamma $ for all $E_\gamma$ in the sequence. Hence, we get $x \notin E_1$, a contradiction. 
We now show that if $x \in E_\beta \setminus E_{\beta+1}$ for some $\beta < \alpha$ then there exists an open ball $B_E (x, \epsilon)$ for some $\epsilon>0$ such that $B_E(x, \epsilon) \cap E_{\beta+1} = \emptyset$. By definition of the sequence of $f$-removed sets the fact that $x \in E_\beta \setminus E_{\beta+1}$ implies that function $f_{\beta}$ defined as $f_{\beta}= f \restriction_{ E_\beta}$ is continuous on $x$, and by definition of continuity of $f_{\beta}$ on $E_\beta \setminus E_{\beta+1}$ we know that there exists an open ball $B_E(x, \epsilon)$ for some $\epsilon>0$ such that $B_E(x, \epsilon) \cap E_{\beta+1} = \emptyset$.
\end{proof}

\vspace{3mm}

Since, as proved by the authors of \cite{collins2009effective}, computing the unique solution of continuous IVP is always possible, it is natural to expect that being able to obtain analytically the unique solution of any given discontinuous IVP should be directly related to the amount of discontinuity for the right-hand term $f$, and consequently to the ordinal number of nonempty levels of the above sequence of $f$-removed sets. Moreover, we would like to obtain the solution within a countable number of steps. Hence, we want to pinpoint some sufficient conditions on $f$ that permit us to restrict our attention to these well-behaved classes of discontinuous systems. We therefore propose the following definition. 

\begin{definition}[Solvable function] \label{def:solvable}
For some $r,m \in \mathbb{N}$, consider a compact domain $E \subset \mathbb{R}^r$ and a function $f: E \to \mathbb{R}^m$. We say that $f$ is \emph{solvable} if it is a function of class Baire one such that for every closed set $K \subseteq E$ the set of discontinuity points of the restriction $f \restriction_K$ is a closed set.
\end{definition}

We say that an IVP involving ODEs is \emph{solvable} when the right-hand term of the ODEs involved is a solvable function. 

The choice of the terminology \emph{solvable} for these right-hand terms is made more clear by the coming Theorem \ref{thm:main}. Indeed this theorem plays a key role on implying that solutions of solvable IVPs can be obtained through transfinite recursion.

\begin{theorem}
\label{thm:lastord}
For some $r,m \in \mathbb{N}$, consider a compact domain $E \subset \mathbb{R}^r$ and a function $f: E \to \mathbb{R}^m$. If $f$ is solvable, then there exists an ordinal $\alpha<\omega_1$ such that $E_{\alpha}= \emptyset$.
\end{theorem}

\begin{proof}
If the set $E_1$ is empty, then the result is trivial. Let us then suppose $E_1 \neq \emptyset$. The hypothesis on function $f$ allows us to apply Lemma \ref{lemma:Baire} to conclude that $E_1$ has empty interior in $E$. Moreover, we know by hypothesis on $f$ that $E_1$ is a closed set and so $E_1$ is nowhere dense in $E$. It follows that $E_1$ is a nonempty, separable complete metric space since it is a closed subset of a separable complete metric space. Therefore, this means that function $f_1 = f \restriction_{E_1}$ is a function of class Baire one defined over a nonempty, separable complete metric space, and hence we know from Lemma \ref{lemma:Baire} that $E_2$ has empty interior in $E_1$. Moreover, by hypothesis, $E_2$ is a closed set and therefore it is nowhere dense in $E_1$. For every ordinal $\beta< \omega_1$, function $f_\beta = f \restriction_{E_\beta}$ is a function of class Baire one since it is the restriction of a function of class Baire one; therefore, we can repeat the above reasoning to show that as a direct consequence of Lemma \ref{lemma:Baire} and the hypothesis on $f$ each set in the sequence of $f$-removed sets on $E$ is nowhere dense in its predecessor (at successor stages). Hence, we know that such sequence is a strictly decreasing sequence, meaning that $E_\beta \subset E_\gamma$ if $\beta > \gamma$ for every $E_\beta$ and $E_\gamma \neq \emptyset$ in the sequence. That means that thanks to Theorem \ref{thm:statio} applied to the sequence of $f$-removed sets on $E$ there exists $\alpha <\omega_1$ such that $E_{\alpha}= \emptyset$ which trivially implies $E_{\beta}= \emptyset$ for all $\beta \geq \alpha$. 
\end{proof}

\vspace{3mm}

It is clear from the above theorem that the definition of solvable functions provides an intuitive way to rank solvable IVPs, where the rank naturally corresponds to the first countable ordinal that leads to the empty set in the sequence of removed sets.

\section{Search method}
\label{sec:search}

Once we have singled out which conditions we must require for the right-hand term $f$, we can present the main tool used to converge to the solution of the IVP. In a similar fashion to the method designed by Denjoy, where the tool to be repeatedly applied was Lebesgue integration, we need to be able to apply such tool for each considered level of the sequence of $f$-removed sets on $E$ until we finally reach the empty set. This is why we created a tool that can be defined for any countable ordinal in a uniform manner. We call the tool $(\alpha)$Monkeys approach in honor of the ten thousand monkey algorithm from \cite{collins2009effective}, since such an algorithm inspires the definition. 

\begin{definition}[$(\alpha)$Monkeys approach]
\label{def:alphamonkeys}

Consider an interval $[a,b] \subset \mathbb{R}$, a compact domain $E \subset \mathbb{R}^r$ for some $r \in \mathbb{N}$ and a right-hand term $f: E \to \mathbb{R}^r$ for an IVP of the form of the one in Equation \ref{eq:pb} with initial condition $y_0$. Let $\{ E_\gamma \}_{\gamma}$ be the sequence of $f$-removed sets on $E$ and let $E_\alpha$ be one set in the sequence for some $\alpha< \omega_1$. We call the \emph{$(\alpha)$Monkeys approach for $(f, y_0)$} the following method: consider all tuples of the form $\left(X_{i, \beta, j }, h_{i, \beta, j }, B_{i, \beta, j }, C_{i, \beta, j }, Y_{i, \beta, j }\right)$ for $i=0, \ldots, l-1$, $\beta < \alpha$, $j=1, \ldots, m_{i,\beta}$, where $h_{i, \beta, j } \in \mathbb{Q}^+$, $l, m_i \in \mathbb{N}$ and $X_{i, \beta, j }$, $B_{i, \beta, j }$, $C_{i, \beta, j }$ and $Y_{i, \beta, j }$ are open rational boxes in $E$. A tuple is said to be valid if $y_0 \in \bigcup_{\beta, j} X_{0, \beta, j }$ and for all $i=0, \ldots, l-1$, $\beta < \alpha$, $j=1, \ldots, m_{i,\beta}$ we have:

\begin{enumerate}
\item Either $\left( B_{i, \beta, j }= \emptyset \right)$ or $\left( \operatorname{cl}(B_{i, \beta, j }) \cap E_\beta \neq \emptyset \; \text{and} \; \operatorname{cl}(B_{i, \beta, j }) \cap E_{\beta+1} = \emptyset  \right)$
\item $ f \restriction_{E_\beta} \left( \operatorname{cl}(B_{i, \beta, j }) \right) \subset C_{i, \beta, j }$;
\item $ X_{i, \beta, j } \cup  Y_{i, \beta, j } \subset B_{i, \beta, j }$;
\item $ X_{i, \beta, j } + h_{i, \beta, j } C_{i, \beta, j }\subset Y_{i, \beta, j }$;
\item $ \bigcup_{\beta, j} Y_{i, \beta, j } \subset \bigcup_{\beta,j} X_{i+1,\beta, j}$;
\end{enumerate}
\end{definition}

We now explain informally the intuition behind such definition. Similarly to the ten thousand monkey algorithm, this definition defines a search method within the space $E$ where the solution of the IVP lives. The search is performed by considering tuples of the form $\left(X_{i, \beta, j }, h_{i, \beta, j }, B_{i, \beta, j }, C_{i, \beta, j }, Y_{i, \beta, j }\right)$ which describe finite sequences of $l$ open sets. Each of these tuples should be thought as an expression of its related finite sequence of $l$ open sets, that is $$\{ \bigcup_{\beta, j} X_{i, \beta, j } \}_{i=0,2, \hdots, l-1}.$$ Each of these open sets is the union of a transfinite collection of open rational boxes. Here we already have the first important difference with the ten thousand monkeys algorithm: while the collection of open rational boxes had to be finite in their case, in our case we allow it to be countably infinite. Indeed, we are considering open rational boxes for each value of $\beta< \alpha$,  while $\alpha$ could be greater than $\omega$. Therefore, some transfinite method is needed in order to consider all the possible tuples that can be described this way within $E$.  

Then, among all the possible tuples considered in this way, we turn our attention to the ones that are \emph{valid}. Stating that one of these tuples is valid means two things: one, that the related sequence starts from a set that contains the initial condition, $y_0 \in \bigcup_{\beta, j} X_{0, \beta, j }$, and two, that the sets in the sequence are concatenated correctly according to the five rules above. These rules are chosen so that the concatenation between the sets is dictated by the action of $f$, but only in a controlled fashion, i.e.\ in a manner that takes care of the portions of the domain where each restriction $f_\beta$ is continuous, for all $\beta < \alpha$. This is clarified by the first item in the above list, whose direct consequence is that $f_\beta$ is continuous on all rational boxes $B_{i, \beta, j }$. This is the second important difference with the ten thousand monkeys algorithm: in the original algorithm the evolution of each of the open rational boxes involved at each step was ruled by the action of a single function, the right-hand term $f$, which was continuous; in our case instead, since $f$ is discontinuous, doing the same could lead to situations in which the evolution of one box could not easily be contained within another box of the same radius. Hence, by trying to use only function $f$ to concatenate the open rational boxes involved at each step, we could end up with trajectories that are too spread away from the actual solution. To solve this complication we decided to give up the condition of using a single function for the concatenation, and instead use a different function $f_\beta$ for each box $B_{i, \beta, j }$. In this way we are sure that the evolution of these rational boxes is always continuous, so that we can control their trajectories while keeping their radius bounded. The tradeoff for this local continuity is that by allowing this stratification we are not sure anymore that every valid tuple contains entirely the solution of the IVP. In other words, we loose the possibility of proving some properties like the ones discussed in Theorem \ref{thm:monkeys} for the ten thousand monkeys algorithm. Nonetheless, by considering valid tuples with smaller and smaller radius, even if neither of them necessarily contains the solution in its entirety, we can still make use of them in order to construct a sequence of continuous functions that eventually converges to the solution in the limit. This sets the premises for the next section, where the method to obtain the solution of the IVP is finally presented in details.

\section{Obtaining the solution} 
\label{sec:main}

We have now all the elements needed to describe the transfinite recursion that obtains analytically the solution of the IVP considered. This is done in the proof of the following theorem. 

\begin{theorem}
\label{thm:main}
Consider a closed interval, a compact domain $E \subset \mathbb{R}^r$ for some $r \in \mathbb{N}$ and a function $f: E \to E$ such that, given an intital condition, the IVP of the form of Equation \ref{eq:pb} with right-hand term $f$ has a unique solution on the interval. If $f$ is a solvable function, then we can obtain the solution analytically via transfinite recursion up to an ordinal $\alpha$ such that $\alpha < \omega_1$. 
\end{theorem}

\begin{proof}
Let $[a,b]$ be the closed interval such that $y:[a,b] \to E$ is the unique solution of the IVP with right-hand term $f$ and initial condition $y_0 = y(a)$. Let $\{ E_\gamma \}_{\gamma}$ be the sequence of $f$-removed sets on $E$. Since $f$ is a solvable function, by means of Theorem \ref{thm:lastord} we know that there exists an $\alpha < \omega_1$ such that $E_{\alpha}= \emptyset$ and $E_{\beta}= \emptyset$ for all $\beta \geq \alpha$. Therefore by transfinite recursion up to $\alpha$ based on repeated application of $f$ we can consider the whole sequence of $f$-removed sets on $E$. We now show how to obtain the solution $y([a,b])$. We first pick a $n \in \mathbb{N}$ and consider a valid tuple of the $(\alpha)$Monkeys approach for $(f, y_0)$ for this value of $n$, i.e.\ we consider a valid tuple with $\left(X_{i, \beta, j }, h_{i, \beta, j }, C_{i, \beta, j }, Y_{i, \beta, j }\right)$ for $i=0, \ldots, l-1$, $\beta < \alpha$, $j=1, \ldots, m_{i,\beta}$ dependent on this fixed $n$. To do so, consider a set $\bigcup_{\beta, j} X_{0, \beta, j }$ such that $y_0 \in \bigcup_{\beta, j} X_{0, \beta, j }$. Then, for all sets $\bigcup_{\beta, j} X_{i, \beta, j }$ for all $i=0, \ldots, l-1$, note that thanks to Lemma \ref{lemma:somealpha} we can select each open rational box $ X_{i, \beta, j }$ so that it satisfies either $\left( X_{i, \beta, j }= \emptyset \right)$ or $\left( X_{i, \beta, j } \cap E_\beta \neq \emptyset \; \text{and} \; X_{i, \beta, j } \cap E_{\beta+1} = \emptyset  \right)$ for all $i=0, \ldots, l-1$, $\beta< \alpha$, $j=1, \ldots, m_{i,\beta}$. Moreover, because every continuous function on a compact domain is uniformly continuous, we can define for all $i=0, \ldots, l-1$, $\beta< \alpha$, $j=1, \ldots, m_{i,\beta}$ the function $\delta_{i, \beta, j }: \mathbb{R}^{+} \to \mathbb{R}^{+}$ to be a modulus of continuity of $f \restriction_{E_\beta}$ on $\operatorname{cl} \left( X_{i, \beta, j } \right) \cap E_\beta$, i.e.\ a function such that $\left\Vert f \restriction_{E_\beta} (x) - f \restriction_{E_\beta} (z) \right\Vert < \delta_{i, \beta, j} ( \left\Vert x - z\right\Vert ) $ for all $x,z \in \operatorname{cl} \left( X_{i, \beta, j } \right) \cap E_\beta$, for all $i=0, \ldots, l-1$, $\beta< \alpha$, $j=1, \ldots, m_{i,\beta}$. By convention, for all $\epsilon >0$, if there are no points $x,z \in\operatorname{cl} \left( X_{i, \beta, j } \right) \cap E_\beta$ such that $\left\Vert x - z\right\Vert< \epsilon $ then we define $\delta_{i, \beta, j}(\epsilon)=\epsilon$. Let us now call $K \in \mathbb{Q}^+$ a rational such that $\max_{x \in E} \left\Vert x \right\Vert < K$. At this point, by taking the partition suffiently small, we can make sure to consider each nonempty open rational box $X_{i, \beta, j}$ with rational $h_{i,\beta,j}$ such that $0 < \operatorname{rad}(X_{i,\beta,j}) < \delta_{i, \beta, j} (\frac{1}{2n}) - K h_{i,\beta,j}$ and such that its $K h_{i,\beta,j}$ neighbourhood has no intersection with $E_{\beta+1}$ and has the same modulus of continuity $\delta_{i, \beta, j}$. Take then each one of these neighbourhoods as the set $B_{i,\beta,j}$ for all $i=0, \ldots, l-1$, $\beta< \alpha$, $j=1, \ldots, m_{i,\beta}$. It follows that $\operatorname{rad}(B_{i,\beta,j}) < \operatorname{rad}(X_{i,\beta,j}) + K h_{i,\beta,j} < \delta_{i, \beta, j} (\frac{1}{2n})$. We then choose open rational boxes $C_{i,\beta,j}$ such that they satisfy $f \restriction_{E_\beta}(\operatorname{cl} (B_{i,\beta,j})) \subset C_{i,\beta,j}$. Note that, by definition of the moduli of continuity, we can choose these boxes so that $\operatorname{rad}\left( C_{i,\beta,j} \right)< \frac{1}{2n}$. From these choices it follows that we can pick open rational boxes $Y_{i,\beta,j}$ satisfying $X_{i,\beta,j}+ h_{i,\beta,j} C_{i,\beta,j} \subset Y_{i,\beta,j}$ and $\operatorname{rad}\left( Y_{i,\beta,j} \right)< \delta_{i, \beta, j} (\frac{1}{2n})$ for all $i=0, \ldots, l-1$, $\beta < \alpha$, $j=1, \ldots, m_{i,\beta}$. Finally, we consider as the set $\bigcup_{\beta, j} X_{i+1, \beta, j }$ a set such that $\bigcup_{\beta, j} Y_{i, \beta, j } \subset \bigcup_{\beta, j} X_{i+1, \beta, j }$ for all $i=0, \ldots, l-1$. It is clear that the tuple described this way is a valid tuple of the $(\alpha)$Monkeys approach for $(f, y_0)$. 

Let us now define two sequences $\{ h_{i, \beta(i), j(i)} \}_{i=0,\ldots, l-1}$ and $\{ t_i \}_{i=0,\ldots, l}$ where $t_0=a$ and $t_i= a+ \sum_{k=0}^{i-1} h_{k, \beta(k), j(k)}$ for all $i=1,\ldots, l$ and a piecewise linear function $\eta_{n}: [a, t_l] \to E$ such that $\eta_{n}(a)=y_0$ and such that for all $i<l$ we have $\eta_{n}(t_i) \in X_{i, \beta(i), j(i) }$ and $\eta_{n}(t)=\eta_{n}(t_i) + (t-t_i)c_{i,\beta(i), j(i)}$ for some $c_{i,\beta(i), j(i)} \in C_{i,\beta(i), j(i)}$, for all $t_{i}< t \leq t_{i+1}$ for all $i=0,\ldots, l-1$. Note that this function is well defined because $ \left\vert t_{i+1} - t_{i}\right\vert = h_{i,\beta(i), j(i)}$ for all $i=0,\ldots, l-1$ and so it follows that $\eta_{n}(t) \in Y_{i,\beta(i), j(i)} \subset \bigcup_{\beta, j} X_{i+1, \beta, j }$ for all $t_{i}< t \leq t_{i+1}$. In other words, we can always choose the sequences in such a way that function $\eta_{n}$ is well defined. Moreover, note that $\eta_{n}'(t) =c_{i,\beta(i), j(i)} \in C_{i, \beta(i), j(i)}$ for all $t_{i}< t < t_{i+1}$, for all $i=0,\ldots, l-1$; note also that since $\eta_{n}(t) \in B_{i,\beta(i), j(i)}$ we have $f_{\beta(i)}(\eta_{n}(t)) \in C_{i,\beta(i), j(i)}$ for all $t_{i}< t < t_{i+1}$, for all $i=0,\ldots, l-1$. Therefore we have $\left\Vert \eta_{n}'(t) - f _{\beta(i)}(\eta_{n}(t)) \right\Vert \leq \operatorname{diam}\left( C_{i,\beta,j} \right) < \frac{1}{n}$ for all $t_{i}< t < t_{i+1}$ such that $\eta_{n}(t) \in E_{\beta(i)}$, for all $i=0,\ldots, l-1$.

Suppose we have just considered a valid tuple of the $(\alpha)$Monkeys approach for $(f, y_0)$ for a fixed value $\bar{n} \in \mathbb{N}$ following the above procedure and we have defined function $\eta_{\bar{n}}: [0,t_l] \to E$ in the way described. Let us indicate this $t_l $ with the symbol $T$. It is clear that we can consider a new valid tuple and a new function $\eta_{n}: [0,T] \to E$ for each value of $n > \bar{n}$ while maintaining the same domain for each function. We can then consider a sequence of the functions $\{ \eta_{n} \}_{n > \bar{n}}$ as defined above, where each function in the sequence is defined based on the valid tuple $\left(X_{i, \beta, j }, h_{i, \beta, j }, C_{i, \beta, j }, Y_{i, \beta, j }\right)$ with $\operatorname{rad}\left( X_{i,\beta,j} \right)< \delta_{i,\beta,j} ( \frac{1}{2n}) $, $\operatorname{rad}\left( B_{i,\beta,j} \right)< \delta_{i,\beta,j} ( \frac{1}{2n}) $ and $\operatorname{rad}\left( C_{i,\beta,j} \right)< \frac{1}{2n}$ for all $n > \bar{n}$, $i=0, \ldots, l-1$, $\beta < \alpha$, $j=1, \ldots, m_{i,\beta}$ as described above. We want to show that such sequence $\{ \eta_{n} \}_{n > \bar{n}}$ is uniformly bounded and equicontinuous. To prove that it is uniformly bounded we need to prove that there exists a constant $R \in \mathbb{R}$ such that $\left\Vert \eta_{n}(t) \right\Vert \leq R$ for all $n > \bar{n}$, for all $a \leq t \leq T$. This is indeed trivial since $\eta_{n}(t) \in \bigcup_{i,\beta, j} X_{i, \beta, j }$ for all $n > \bar{n}$, for all $a \leq t \leq T$ and each open rational box $X_{i, \beta, j } \subset E$ for all $n > \bar{n}$, $i=0, \ldots, l-1$, $\beta< \alpha$, $1 \leq j \leq m_{i,\beta}$. For equicontinuity it is enough to prove that there exists a constant $M \in \mathbb{R}$ such that $\left\Vert \eta_{n}(\tilde{t}) - \eta_{n}(t)  \right\Vert \leq M \left\vert t - \tilde{t}\right\vert $ for all $n > \bar{n}$, for all $a \leq t, \tilde{t} \leq T$. The existence of $M$ follows from the fact that the sequence is uniformly bounded together with the fact that $ \left\Vert \eta_n'(t) \right\Vert < K$ for all $n > \bar{n}$, for almost all $a \leq t \leq T$. Therefore, since the sequence is uniformly bounded and equicontinuous, we can apply Ascoli's theorem together with Theorem \ref{thm:uniform} in order to conclude that the sequence $\{ \eta_{n} \}_{n > \bar{n}}$ has a subsequence $\{ \eta_{n(u)} \}_{u > \bar{n}}$ that converges uniformly on $[a,T]$ to a continuous function $\eta: [a, T] \to E$. Moreover, thanks to Theorem \ref{thm:differsequence}, we know that function $\eta$ is differentiable almost everywhere on $[a,T]$. Note that by taking limit of $n \to \infty$ we have $l \to \infty$ and $\operatorname{rad}\left( B_{i,\beta,j} \right), \operatorname{rad}\left( C_{i,\beta,j} \right) \to 0$ and $h_{i,\beta,j} \to 0$ for all $i=0, \ldots, l-1$, $\beta< \alpha$, $j=1, \ldots, m_{i,\beta}$. Therefore the inequality $\left\Vert \eta_{n}'(t) - f \restriction_{E_{\beta(i)}}(\eta_{n}(t)) \right\Vert \leq \operatorname{diam}\left( C_{i,\beta,j} \right) < \frac{1}{n}$ for all $t_{i}< t < t_{i+1}$ such that $\eta_{n}(t) \in E_{\beta(i)}$, for all $i=0,\ldots, l-1$, leads to the equation $ \eta'(t) = f(\eta(t))$ for almost all $t \in [a,T]$. Since $\eta(a)= y_0$, continuity and unicity of the solution of the IVP imply $\eta(t)=y(t)$ for all $t \in [a,T]$. Specifically this means that any convergent subsequence converges to the same function, which is precisely the solution $y$ of the IVP. 

Finally, to obtain the solution over the whole domain $[a,b]$, it is sufficient to consider the initial valid tuple of the $(\alpha)$Monkeys approach for $(f, y_0)$ as described above but in such a way that, when defining the sequence of times $\{ t_i \}_{i=0,\ldots, l}$ we have $t_l \geq b$ and then take $T=b$. This is always possible due to the definition of a valid tuple and the fact that $f$ is bounded within $E$.

\end{proof}

\section{Examples}
\label{sec:examples}

In this section we construct many examples of solvable IVPs involving discontinuous ODEs. We start with a simple case that will be later used as a building block to derive more and more complex examples. We stress once again here that none of the most used methods proposed in analysis is able to obtain as a limit the solution of the following example, since every argument based on existence of a fixed point in the space of bounded, continuous functions can not be applied due to the discontinuities in the domain. 

\begin{example}[A discontinuous IVP with a unique solution]
\label{ex:firstex}
As the simplest case of discontinuous IVP, we consider the following example: let $E= [-5,5] \times [-15,15]$ and define the function $f: E \to E$ as 
$$f(x,z)=\left\{ \begin{array}{ll}
\left(1,2x \sin{\left( \frac{1}{x}\right)} - \cos{\left( \frac{1}{x}\right)}\right) & \mbox{ if } x\neq 0 \\
  (1,0) & \mbox{ otherwise} \end{array} \right.$$
   It is easy to see that $f$ is a function of class Baire one, i.e.\ it is the pointwise limit of a sequence of continuous functions. Note also that in this case the set of discontinuity points of function $f$ on $E$ is the closed set $D_f= \{ (0,z)$ for $z \in [-15,15] \}$.  Then consider the following IVP, with $y: [-2,2] \to \mathbb{R}^2$ and $y_0=\left( -3, 9 \sin{ \left(-\frac{1}{3}\right)}\right)$:

\begin{equation}
\label{eq:derivODE2dimens}
\begin{cases}
y'(t)= f(y(t))\\
y(-2)= y_0
\end{cases}
\end{equation}

It is easy to verify that the solution of such a system is unique, and it is for the first component: $y_1(t)= t- 1$, and for the second component $y_2 (t)= (t-1)^2 \sin{\left(\frac{1}{t-1}\right)}$ for $t \neq 1$ and $y_2(1)=0$. Therefore the solution $y: [-2,2] \to \mathbb{R}^2$ is differentiable and can be expressed as the unique solution of the IVP above with right-hand side $f$ discontinuous on $E$. Note that the only discontinuity of $f$ encountered by the solution is the point $(0,0)$, i.e. $D_f \cap y([-2,2])= (0,0)$. Moreover, it is easy to see that the right-hand term $f$ is a solvable function. Indeed, $f$ is trivially a function of class Baire one, while $f \restriction_{D_f}$ is a constant function identically equal to $(1,0)$ and therefore continuous everywhere on its domain. This also implies $E_2=\emptyset$, where $E_2$ is the second set in the sequence of $f$-removed sets on $E$. 
\end{example}



The construction of this simple example (and of the solution of this classical exercise) is based on the well-known fact that the real function $f(x)=x^2 \sin{ \left( \frac{1}{x} \right)}$ if $x\neq 0$ and $f(0)=0$ is differentiable over $[0,1]$ and its derivative is bounded and discontinuous in $0$. Moreover, we avoided some problems that arise for mono-dimensional ODEs with null derivative by introducing a \emph{time} variable $y_1$ whose role is to prevent the system from stalling and ensure the unicity of the solution. 

The concept behind such an example can be easily generalized. 

\begin{example}[Converting complex derivatives into complex IVPs] \ \\
\label{ex:derivODE}
Whenever we consider a differentiable function $g: [a,b] \to \mathbb{R}$ such that $g(a)=g_0$ and with bounded derivative $g': [a,b] \to \mathbb{R}$ we can obtain such function as a solution of an IVP of the type of the one in Equation \ref{eq:pb} by constructing a system as the following: 

\begin{equation}
\label{eq:derivODE}
\begin{cases}
y'_1(t)= 1\\
y'_2(t)= g' (y_1(t)) \\
\end{cases}
\begin{cases}
y_1(a)= a \\
y_2(a)= g_0
\end{cases}
\end{equation}
\end{example}

This consideration allows us to construct examples for which the set of discontinuity points of the right-hand term on the domain is more and more sophisticated. Indeed, this can be done by simply selecting more and more elaborate discontinuous derivatives in Equation \ref{eq:derivODE}. 

But before presenting some of these cases, we make use of the technique introduced by the above example to construct an IVP with unique solution for which the function used as right-hand term is not a solvable function. Specifically, this means that the class of solvable systems is a proper subclass within the discontinuous systems with unique solutions. 

\begin{example}[Non solvable IVP with unique solution]
\hfill We introduce an \newline IVP of the type of the one in Equation \ref{eq:pb} with unique solution and such that the set of discontinuity points of its right-hand term function is not a closed set. The latter means that our transfinite method proposed in the previous section can not be applied to this particular case. In order to describe this case, we need the notion of \emph{nowhere monotone} function. Intuitively, these are functions over the reals whose sign does change in every interval of their domain. 

\begin{definition}[Nowhere monotone]
Let $f : \mathbb{R} \to \mathbb{R}$, we say that $f$ is \emph{nowhere monotone} if there is no interval $I \subseteq \mathbb{R}$ such that $f(x) \geq 0$ for all $x \in I$ or $f(x) \leq 0$ for all $x \in I$.
\end{definition} 

The reason to introduce here the notion of nowhere monotone functions is that there exist nowhere monotone differentiable functions with bounded derivatives. The existence of such functions has been proved by Weil \cite{WE76} as an application to the Baire cathegory theorem. Many examples of such derivatives have been described in literature, whose construction is often quite technical. As examples the reader can consider the one presented by Pereno (1897) as reproduced in \cite{Hobs21} (pages 412-421), as well as the one in \cite{Kat74}. 

We can now state the following lemma:

\begin{lemma}
\label{lemma:notclose}
Let $f :\mathbb{R} \to \mathbb{R}$ be a differentiable, nowhere monotone function. Then the set of discontinuity points of its derivative is not a closed set. 
\end{lemma}

\begin{proof}
As usual let us call $D_{f'}$ the set of discontinuity point of the derivative of $f$. Let us define the set $N= \{ x : f'(x)= 0\}$ and let $S$ be the complement of $D_{f'}$, i.e.\ $S$ is the set of points in which function $f'$ is continuous. We first show that $S \subset N$. Indeed, pick any point $x_0 \in S$ and assume that $f'(x_0) \neq 0$. Then, by definition of continuity, there exists an $\epsilon >0$ such that for all $x \in (x_0 -\epsilon, x_0+\epsilon)$ we have $f'(x) \neq 0$. This implies a contradiction, since $f$ is nowhere monotone. Therefore, $S \subset N$. We now show that both $N$ and its complement $\mathbb{R} \setminus N$ are dense. Since $f$ is nowhere monotone, for each interval $I$ there exist at least two points $x_1$,$y_1 \in I$ with $x_1 < y_1$ such that $f(x_1) < f(y_1)$ and two points $x_2$,$y_2$ with $x_2 < y_2$ such that $f(x_2) > f(y_2)$. Then, by the mean value theorem, this implies that there exist two points $z_1 \in (x_1, y_1)$, and $z_2 \in (x_2, y_2)$ such that $f'(z_1)>0$ and $f'(z_2)<0$. That means that $z_1, z_2 \in \mathbb{R} \setminus N$, and hence $\mathbb{R} \setminus N$ is dense. Moreover, since $f'$ is a derivative, due to Darboux theorem we also know that there is a point $z$ between $z_1$ and $z_2$ such that $f'(z)=0$. This implies that $N$ is dense. 
Note now that since $S \subset N$, we have $\mathbb{R} \setminus N \subset D_{f'}$. Therefore, because $D_{f'}$ contains a dense set, $D_{f'}$ is also dense. This implies that either $D_{f'}$ is not a closed set, either $D_{f'}=\mathbb{R}$. Since $f'$ is a derivative, it is a function of class Baire one. Therefore, due to Lemma \ref{lemma:Baire} we know that $D_{f'}$ has empty interior. Hence, $D_{f'} \neq \mathbb{R}$ and so $D_{f'}$ is not a closed set. 
\end{proof}

\vspace{3mm}

We note that the above proof applies as well to functions defined over any closed subinterval of $\mathbb{R}$. Hence, if starting from the bounded derivative of any differentiable, nowhere monotone function, we build an IVP with the technique of Example  \ref{ex:derivODE}, then its solution is unique while the set of discontinuity points of the right-hand term is not a closed set due to the above lemma. 
\end{example}

At this point we start discussing more complicated examples constructed with the technique of Example \ref{ex:derivODE} and which are based on the building block presented in Example \ref{ex:firstex}. 


\begin{example}[Cantor set of discontinuities]
\label{ex:secondex}

Let us now construct the case of a right-hand function $f$ for an IVP defined on a compact domain $E$ for which $D_f$, the first set in the sequence of $f$-removed sets on $E$, is homeomorphic to the Cantor set. This case is considerably more complex compared to the previous one, from a technical standpoint, but it is theoretically based on the same concept. Indeed, the idea is to make use of the discontinuous derivative just seen for the previous case, and copy it inside the Cantor set. This is done in a similar fashion to what happens when defining Volterra's function \cite{Br08}.

Concretely, consider the Cantor set $C$. Then  define the interval $I_{m,n}=(a_{m,n}, b_{m,n})$ as the $m$-th interval removed from $[0,1]$ at the $n$-th step of the construction of $C$, where $n \in \mathbb{N}_0$. Note that at each step $n \in \mathbb{N}_0$ there are $2^{n}$ removed intervals, i.e.\ $m \in \{ 1, 2,  ..., 2^{n} \}$. For each of the $I_{m,n}$ interval, define function $g_{m,n}: [0,1] \to \mathbb{R}$ by cases as:

\begin{equation}
\begin{cases}
g_{m,n}(x)= (x-a_{m,n})^2 \sin \left( \frac{1}{x-a_{m,n}}\right) & \text{  if  } x \in (a_{m,n}, \bar{x}_{m,n}] \\
g_{m,n}(x)= (\bar{x}_{m,n}-a_{m,n})^2 \sin \left( \frac{1}{\bar{x}_{m,n}-a_{m,n}}\right) & \text{  if  } x \in [\bar{x}_{m,n}, b_{m,n} - \bar{x}_{m,n}] \\
g_{m,n}(x)= -(x-b_{m,n})^2 \sin \left( \frac{1}{x-b_{m,n}} \right) &  \text{  if  } x \in [b_{m,n} - \bar{x}_{m,n}, b_{m,n}) \\
g_{m,n}(x)= 0 & \text{ otherwise}
\end{cases}
\end{equation}

where the point $\bar{x}_{m,n}$ is the greatest point in $I_{m,n}$ such that: $\bar{x}_{m,n} < \frac{a_{m,n}+b_{m,n}}{2}$ and such that: $$ 2(\bar{x}_{m,n}-a_{m,n}) \sin{ \left( \frac{1}{\bar{x}_{m,n}-a_{m,n}}\right) } - \cos{ \left( \frac{1}{\bar{x}_{m,n}-a_{m,n}} \right)}=0.$$ 

Note that each of these functions is differentiable and that the derivative is discontinuous in $a_{m,n}$ and $b_{m,n}$. Indeed, the expression of the derivative is the following: 

{\small
\begin{equation}
\begin{cases}
g'_{m,n}(x)= 2(x-a_{m,n}) \sin{ \left( \frac{1}{x-a_{m,n}} \right)} - \cos{\left( \frac{1}{x-a_{m,n}} \right)} & \text{  if  } x \in (a_{m,n}, \bar{x}_{m,n}]\\
g'_{m,n}(x)= 0 & \text{  if  } x \in [\bar{x}_{m,n}, b_{m,n} - \bar{x}_{m,n}] \\
g'_{m,n}(x)= -2(x-b_{m,n}) \sin{\left( \frac{1}{x-b_{m,n}} \right)} + \cos{\left( \frac{1}{x-b_{m,n}} \right) } & \text{  if  } x \in [b_{m,n} - \bar{x}_{m,n}, b_{m,n}) \\
g'_{m,n}(x)= 0 & \text{ otherwise}
\end{cases}
\end{equation}
}

At this point we can construct a function $g: [0,1] \to \mathbb{R}$ in the following way: 

\begin{equation}
g(x)= \displaystyle\sum_{n=0}^{\infty} \displaystyle\sum_{m=1}^{2^n} g_{m,n}(x)
\end{equation}

First note that due to Ascoli's theorem the sequence $\{ G_{k} \}_k$ of continuous functions defined as $G_k (x)= \displaystyle\sum_{n=0}^{k} \displaystyle\sum_{m=1}^{2^n} g_{m,n}(x)$ for all $x \in [0,1]$, for all $k \in \mathbb{N}$ converges uniformly to the continuous function $g$. We now discuss the differentiability of function $g$ by means of the following lemma:

\begin{lemma}
The function $g$ defined as above is differentiable and its derivative $g'$ is discontinuous only on the Cantor set $C$.
\end{lemma}

\begin{proof}
Note that the sequence $\{ R_{k} \}_k$ of functions defined as $R_k (x)= \displaystyle\sum_{n=0}^{k} \displaystyle\sum_{m=1}^{2^n} g'_{m,n}(x)$ for all $x \in [0,1]$, for all $k \in \mathbb{N}$, converges uniformly to function $r =\displaystyle\sum_{n=0}^{\infty} \displaystyle\sum_{m=1}^{2^n} g'_{m,n}$. Therefore, the differentiable limit theorem presented in Theorem \ref{thm:difflimit} tells us that function $g$ is differentiable with derivative $g'=r$. We now want to show that $g'$ is discontinuous exactly at each point of the Cantor set $C$. First, note that if $x \in I_{m,n}$ for any $m,n$, it is easy to check that $g'(x)=g'_{m,n}(x)$ since the intervals $I_{m,n}$ are pairwise disjoint, and hence $g'$ is continuous on $I_{m,n}$ since $g'_{m,n}$ is. We are then left with analyzing the continuity of function $g'$ on $C$. Recall that each point in the Cantor set $C$ is an accumulation point both for $C$ and for its complement $[0,1] \setminus C$. Moreover, note now that, for each $I_{m,n}$ there is at least one point $x \in I_{m,n}$ such that $g'(x) = 1$. This indeed follows from the above expression of each function $g'_{m,n}$. Therefore, for all points $c \in C$, we can build a sequence of points $\{ x_n \}_n$ such that $g'(x_n) = 1$ for all $n \in \mathbb{N}$ and such that the sequence converges to $c$. In the same way, since for all $x \in C$ we have $g'(x)=0$, we can also build a sequence of points $\{ y_n \}_n$ such that that  $g'(y_n) = 0$ for all $n \in \mathbb{N}$ and such that the sequence converges to $c$. This implies that function $g'$ is discontinuous on $c$ for all $c \in C$, proving the lemma. \end{proof}

\vspace{3mm}

At this point, by using function $g'$ just defined as the derivative involved with the construction of an IVP of the type of Equation  \ref{eq:derivODE}, it is clearly possible to construct an IVP on a compact domain $E \subset \mathbb{R}^2$ with a right-hand term $f$ for which we have that $D_f$, the first set in the sequence of $f$-removed sets on $E$, is homeomorphic to the Cantor set. Picking any $K>0$ such that $\max_{t \in [0,1]} ( \left\vert g(t) \right\vert ,\left\vert g'(t) \right\vert )< K$, a way to do so is to consider the compact domain $E= [0,1] \times [-K,K]$, the right-hand term $f: E \to E$ defined as $f(x,z)= (1, g'(x))$ for all $(x,z) \in E$, the initial condition $y_0=(0,g(0))$ and the IVP on $[0,1]$: 

\begin{equation}
\begin{cases}
y'(t)= f(y(t))\\
y(0)= y_0
\end{cases}
\end{equation}

Note that in this way the solution $y$ of this IVP is unique, while it is easy to see that the right-hand term $f$ is a solvable function. 

\end{example}

Concerning the above example, one interesting observation can be made. 

\begin{remark}
If instead of using the Cantor set we choose to select as set of discontinuity points the fat Cantor set, which is a well known modified version of the Cantor set in which the middle nth is removed from the unit interval at the nth step of its construction (instead of the middle third) then we end up with a solution of the IVP whose derivative has a set of discontinuity points of positive Lebesgue measure. This implies that this derivative is not Riemann integrable, and that one essential condition of the Caratheodory method for solving discontinuous ODEs is not met. Nonetheless, our method can still successfully recover the solution also for these pathological cases. 
\end{remark}

Moving on to the next example, we set the premises to illustrate the fact that, for each countable ordinal $\alpha < \omega_1$, we can construct an example of an IVP with solvable right-hand term $f$ such that its sequence of $f$-removed sets satisfies $E_\alpha \neq \emptyset$. Note that all the examples presented so far satisfy $E_1 \neq \emptyset$ and $E_2 = \emptyset$. Indeed, the difference between Example \ref{ex:firstex} and Example \ref{ex:secondex} relies on the form of the set of discontinuity points $E_1$, but the method we can apply to obtain the solution is still the same and it is based on the application of the $(2)$Monkeys approach since $E_2= \emptyset$. We now construct one example that also satisfies $E_2 \neq \emptyset$. 

\begin{example}[Discontinuities for $f_{D_f}$]

\label{ex:thirdex}

Preserving the definitions and the notation introduced in Example \ref{ex:secondex}, this time we define the function $g: [0,1] \to \mathbb{R}$ in the following way: 

\begin{equation}
\begin{cases}
g(x)= x^2 \sin \left( \frac{1}{x} \right) + \displaystyle\sum_{n=0}^{\infty} \displaystyle\sum_{m=1}^{2^n} g_{m,n}(x) & \text{  if  } x \neq 0 \\
g(x)= 0 & \text{ otherwise}
\end{cases}
\end{equation}

It is easy to see that the previous analysis still applies to this case, meaning that function $g$ is continuous and differentiable, with derivative $g'$ that is discontinuous on the Cantor set $C$. Nonetheless, if we consider the restriction of such function on $C$ we have $g \restriction_{C} (x)= x^{2} \sin \left( \frac{1}{x} \right)$ and its derivative will be $g' \restriction_{C} (x)= 2x \sin{ \left( \frac{1}{x} \right) } - \cos{ \left( \frac{1}{x} \right) }$ which is bounded and discontinous on $x=0$. Therefore, this implies that by using function $g'$ just introduced as the derivative involved with the construction of an IVP of the type of the one in Equation \ref{eq:derivODE} we can obtain an IVP defined over a compact domain $E$ with a solvable right-hand term $f$ such that the first set $E_1$ in the sequence of $f$-removed sets on $E$ is homeomorphic to the Cantor set while the second set $E_2 \neq \emptyset$.

\end{example}

The technique just presented can be extended naturally in order to build, for any $\alpha < \omega_1$, right-hand terms for which we have $E_\alpha \neq \emptyset$. We do not decribe here the whole procedure due to the cumbersome details involved, but instead we provide the intuition, since all the conceptual elements required for the extension have been already provided by the last example. Indeed, by adding extra terms of the type $ (x-x_0)^{2} \sin \left( \frac{1}{x-x_0} \right)$ for some $x_0 \in C$  in the definition above we could generate more discontinuity points for $g' \restriction_{C}$ in the same way it was just done for $x=0$. Specifically, if we take some technical care to sum an uncountable number of them (in a similar way to what was already done before using the intervals $I_{m,n}$ in order to have discontinuities in the whole Cantor set) we can obtain a set of discontinuity points for $g' \restriction_{C}$ that is homeomorphic to the Cantor set. This is a consequence of the fact that every closed uncountable subset of the Cantor set $C$ is homeomorphic to the Cantor set itseself. This leads to an IVP defined over a compact domain $E$ with a solvable right-hand term $f$ such that $E_2$, the second set in the sequence of $f$-removed sets on $E$, is homeomorphic to the Cantor set. At this point it is easy to see that the whole process can be straightforwardly repeated any countable number of times by adding the same countable number of extra terms defined by cases, where each one of these terms represents a new uncountable set of discontinuities included in the previously added one. Hence, in this way we can construct, for any $\alpha < \omega_1$, an example of a right-hand term $f$ that satisfies $E_\alpha \neq \emptyset$. 

Another alternative way to obtain the same extension is to construct examples with the technique of Example \ref{ex:derivODE} applied to any of the known complexity ranks for differentiable functions. Indeed, several \emph{differentiabiliy ranks} for measuring descriptive complexity of differentiable functions have been introduced in literature. These ranks can be used for the purpose of providing in a structured way more and more complex functions to play the role of function $g$ in Equation \ref{eq:derivODE}: some of these ranks are, for istance, the Kechris-Woodin rank or the Zalcwasser rank, which can be found in \cite{kechris1986ranks} and \cite{Za30} respectively. The relations between all the existing notions are investigated in \cite{Ra91} and \cite{ki1997denjoy}.

However, independently from which of these two routes is taken, the key concept behind the creation of highly elaborate examples is the same: one elementary discontinuous derivative (such as the one in Example \ref{ex:firstex}) is used as a building block that is then rescaled and concatenated into smaller and smaller intervals converging to a new discontinuity point. It is then clear that the nature of the discontinuity on this point of the new function obtained this way would be strictly more complex than the nature of the discontinuities featured in the example used as an elementary block. Then, by using the function obtained this way as a new elementary block to rescale and concatenate, the process can continue in a fractal-like iterative fashion. 

Therefore, our study establishes the premises for ranking discontinuous IVPs depending on the complexity required to solve them and foretells the development of a related hierarchy. Up to that point, all examples were mostly obtained from integrating various derivatives. Therefore, it might seem logical to directly derive the ranking of these dynamical systems from the above mentioned differentiability ranks. However, IVPs solving is a more general problem than integration, and therefore similar rankings would necessary be incomplete for our tasks. Indeed, this approach does not suit our purpose, since it only characterizes a limited subclass of systems, i.e.\ the systems yielded by application of the trick introduced with Equation \ref{eq:derivODE}. Despite being a relevant and insightful subclass, this approach "from below" fails to exhaust the generality of the problem. By using the sequence of removed sets to represent the ranking, we instead propose an approach "from above" which builds from the commonalities of those examples and extends beyond them to a more complete analysis that is not just tailored on derivatives defined over the reals. 

The following section presents an example of an interesting IVP that has the power of solving the halting problem for Turing machines and that is constructed using a specific solvable right-hand term which is not directly built from a discontinuous derivative in the sense of Example \ref{ex:derivODE}.

\section{The halting set and Turing jumps: a building block for hypercomputation}
\label{sec:jumps}

As we already explained, classes of dynamical systems of ODEs have been largely used to characterize complexity for discrete models of computation, where the main connecting tool used in this context has been to design analog dynamics that could simultate the action of Turing machines' transiction functions. In this spirit, the purpose of this section is to show one main utility of solvable IVPs by outlining that their solutions can be used to simulate transfinite computations. By definition, a transfinite computation is a computation that may produce as output any set in the hyperarithmetical hierarchy. This implies that the maximum number of steps for a transfinite computation is bounded by the first nonrecursive ordinal $\omega_1^{CK}$. Following the parallel established by \cite{TAMC06}, where polynomial ODEs where shown to correspond to computable analysis, the latter would mean that solvable ODEs correspond (in the same sense) to transfinite computing.

\subsection{The halting set.}

As a starting point to prove this equivalence, we are going to construct an example of a solvable IVP whose solution assumes, at an integer time, a real value $\mu$ encoding the halting problem for Turing machines. It is important to stress out that, in this example, the capability of the system to reach such a value does not come from nonrecursiveness of initial datas and setting, but instead naturally arises from the discontinuity of the right-hand term. Indeed, the IVP considered in the example has a rational initial condition and a right-hand term defined over a two dimensional domain that is computable everywhere exept a straight-line, where it is discontinuous. We show that this simple structure for the set of discontinuity points of a solvable right-hand term is enough for obtaining a solution that at an integer time assumes a noncomputable value. 

Given a G\"{o}del enumeration of Turing machines, we define the halting problem as the problem of deciding the halting set $H=\{ e : M_e (e) \perp \}$ where $ M_e (e) \perp$ means that machine represented by natural $e$ halts on input $e$. We consider a one-to-one total computable function over the naturals $h:\mathbb{N} \to \mathbb{N}$ that enumerates such a set. It is known that any such function enumerating a noncomputable set naturally generates a noncomputable real number \cite{PORI17}. The following definition expresses this:

\begin{definition}
\label{def:noncompu}
Let $h:\mathbb{N}_0 \to \mathbb{N}$ be a one-to-one computable function such that $h(i)>0$ for all $i \in \mathbb{N}_0$ and such that it enumerates a noncomputable set $A$. Then the real number $\mu$ defined as: 
\begin{equation} \label{def:rel}
\mu= \sum_{i=0}^{\infty} 2^{- h(i)}
\end{equation}
is noncomputable.
\end{definition}

Note that in this way we always have $0<\mu <1$. We now describe a bidimensional dynamical system that generates the real number $\mu$ associated in the sense of the definition above to function $h$ enumerating the halting set. This is illustrated by the following theorem:

\begin{theorem}
\label{thm:HaltingP}
Let $E=[-5,5] \times [-5,5]$. There exists a solvable IVP with unique solution $y: [0,5] \to E$, rational initial condition and right-hand term computable everywhere on $E$ except a straight-line, with  $y_2 (5)=\mu$. 
\end{theorem}

\begin{remark}
The spirit of the construction of such an example is inspired by the technique used in \cite{gracca2008boundedness}, where the solution of the IVP considered is stretched in a controlled manner so that it grows infinitely approaching a fixed noncomputable time. In our case instead, with the above example, we are replacing their indefinite growth with a bounded one, yielding a finite convergence for the solution. This introduces many complications, and the fact that we want to guarantee differentiability for the solution is a true difficulty. 
\end{remark}

\begin{proof}
 
We start presenting the proof by first explaining here informally the ideas involved. Then, once the intuition is clear, we complete the proof by filling the missing technical details. 

Intuitively, we first discretize time by introducing specific time slots in which both components of the solution, $y_1$ and $y_2$, have a well-defined behaviour. Specifically, we require the first component, which is negative, to increase by a factor of $2$ in each of these time slots, converging to zero. Instead, the second component, which is positive, is required to incrementally converge to the real $\mu$ by adding to itself the quantity $2^{-h(i)}$ on the $i$-th time step. We need two components because we want the right-hand term to be computable outside its set of discontinuity points. This is achievable in this way since indeed it is possible to computably implement the correct derivative for $y_1$ in each time slot by only going through the enumeration described by function $h$ while looking at the value of the second component $y_2$ in that slot. Then, the existence and continuity of the solution is granted by designing an infinitely countable sequence of time slots that converges suitably. 

We now informally describe the behavior of the two components of the solution $y=(y_1, y_2)$. The formal details on how to fully construct the IVP that ensures such behaviour are then given later on in the proof. We first define the function that represents the discretized time evolution of the dynamical system:

\begin{definition}          
\label{def:slowtime}
Let $h:\mathbb{N}_0 \to \mathbb{N}$ be a one-to-one computable function such that $h(i)>0$ for all $i \in \mathbb{N}_0$ and such that it enumerates the halting set $H$. Define the function $\tau :\mathbb{N}_0 \to \mathbb{N}$ to be the total computable function such that:
\begin{equation}
\tau (i) =
\begin{cases}
2^{-h(i)/2} & \text{if $h(i) < i$} \\
2^{- i/2}  & \text{if $h(i) \geq i$}
\end{cases}
\end{equation}
\end{definition}

\noindent note that the above definition implies $ 2^{- i/2} \leq \tau(i) \leq 1$ for all $i \in \mathbb{N}_0$. Moreover, the quantity $\tau^*= \sum_{i=0}^{\infty} \tau (i)$ is finite since we have:
\begin{equation}
\begin{split}
\tau^* < & \sum_{i=0}^{\infty} 2^{- h(i)/2} + \sum_{i=0}^{\infty} 2^{- i/2} \\
 & < \mu + \frac{1}{(1-1/\sqrt{2})} = \mu + 2 + \sqrt{2} < 5.
\end{split}
\end{equation}

As we will soon show, this quantity $\tau^*$ is going to represent the time required for the solution to reach the noncomputable value $\mu$. 

\begin{remark} The reason for measuring time steps with Definition \ref{def:slowtime}, instead of directly exploiting the construction of $\mu$ via Definition \ref{def:noncompu}, is technical. The intuition behind it is that we want time to evolve slowly enough when compared to the increasing rate of the solution. This consideration takes care of one of the construction's main difficulties: the solution's differentiability at time $\tau^*$, when the derivative is discontinuous. 
\end{remark}

Let us now proceed to analyze the behaviour of the solution $y$. For the first component $y_1$ we have a dynamic given by a function $f_1$ such that, for all $i \in \mathbb{N}_0$, if we have:

\begin{equation}
\begin{cases}
y'_1(t) = f_1 (y_1(t)) &\forall t \in [0,\tau(i)] \\
y_1 (0) = -2^{-i}
\end{cases}
\end{equation}

\noindent then we have $y_1(\tau(i))= -2^{-(i+1)}$.  In other words, we require $y_1$ to be an increasing function such that at every time step $\tau(i)$ its value increases by a factor of $2$, converging then to $0$ as time converges to $\tau^*$. To achieve this goal we require $f_1$ to be autonomous, with no explicit dependence on time. 


For the second component $y_2$ we have a dynamic given by a function $f_2$ such that, for all $i \in \mathbb{N}_0$, if we have:

\begin{equation}
\begin{cases}
y'_2 (t) = f_2 (y_1(t)) & \forall t \in [0,\tau(i)] \\
y_1 (t) \in [-2^{-i},  -2^{-(i+1)}] &  \forall t \in [0,\tau(i)] \\
y_2 (0) = \sum_{m=0}^{i} 2^{-h(m)}
\end{cases}
\end{equation}

\noindent then we have $y_2(\tau(i))= \sum_{m=0}^{i+1} 2^{-h(m)}$. In other words, we require $y_2$ to be an increasing function such that at every time step $\tau(i)$ its value increases of the quantity $2^{-h(i+1)}$, converging then to $\mu$ as time converges to $\tau^*$. 


A careful analysis proves that such a solution $y$ is indeed differentiable at time $\tau^*$ and that the derivative of both components exists and equals zero at such time. Moreover, we have that $y_1(\tau^* )=0$ and $y_2(\tau^* )=\mu$. At this point, forcing the dynamic to remain constant for the remaining time is sufficient to obtain $y_2 (5)=\mu$. This yields the desired outcome since the IVP has reached, at a computable time, a noncomputable value that encodes the halting problem. 

Now that we have described the ideas and the key concepts of the proof, we can proceed by providing the technical details that are missing. 

To construct the IVP we consider the time interval $[0,5]$ and the domain $E=[-5,5] \times [-5,5]$. We have $f: E \to E$ and we and consider the IVP:  

\begin{equation}
\label{eq:alphaIVP}
\begin{cases}
y'(t)= f(y(t))= (f_1 (y(t)), f_2 (y(t))) & \forall t \in [0,5] \\
y(0) = (-1, 2^{-h(0)})
\end{cases}
\end{equation}

Function $f_1$ is constructed in such a way that its action depends only on the value of its first component. It is defined by cases. First we require $f_1(x_1, x_2)=1$ if $x_1 \in [-5,-1]$ and $f_1(x_1, x_2)=0$ if $x_1 \in [0,5]$. Then we construct it piecewisely on intervals of the form $[- 2^{-i},- 2^{-(i+1)}]$ for all $i \in \mathbb{N}_0$. In each of these intervals, the function will be piecewise linear. More precisely, if $x_1 \in [- 2^{-i},- 2^{-(i+1)}]$ for some $i \in \mathbb{N}_0$ we define:

\begin{equation}
\label{eq:haltderivative}
f_1(x_1,x_2)= 
\begin{cases}
(k(i)2^{i+2}-1) x_1 + 4 k(i)   & \text{            if   } x_1 \in [- 2^{-i}, -3 \cdot 2^{-(i+2)}] \\
- (k(i)2^{i+2}+1) x_1 -2 k(i)  & \text{             if   } x_1 \in [ -3 \cdot 2^{-(i+2)}, - 2^{-(i+1)}]
\end{cases}
\end{equation}

where $k(i)>0$ is a computable real that we can select in such a way that, if we have $y_1(t)=- 2^{-i}$ for any $t \in [0,5-\tau(i)]$ then $y_1(t+\tau(i))=- 2^{-(i+1)}$. 

Note that in this way $f_1(x_1,x_2)$ is clearly always computable whenever $x_1 \neq 0$. Indeed it is easy to check that, for all $x_2 \in [-5,5]$ and for all $i \in \mathbb{N}_0$, we have $f_1(-3 \cdot 2^{-(i+2)},x_2) = k(i)+ 3 \cdot 2^{-(i+2)}$ and $f_1(-2^{-i},x_2) =2^{-i}$.

We now need to show that such $k(i)$ indeed exists and is computable. By integrating the above expression within a time length of $\tau(i)$ we obtain:

\begin{equation}
\begin{split}
\tau(i) = & \int_{-2^{-i}}^{-3 \cdot 2^{-(i+2)}} \frac{d x_1}{(k(i)2^{i+2}-1) x_1 + 4 k(i)} \\
 & - \int_{-3 \cdot 2^{-(i+2)}}^{-2^{-(i+1)}} \frac{d x_1}{(k(i)2^{i+2}+1) x_1 + 2 k(i)} 
\end{split}
\end{equation}

from which it follows that: 

\begin{equation}
\begin{split}
\tau(i) = & \frac{1}{ k(i)2^{i+2}-1 } \ln \left( k(i) 2^{i}+ 3/4 \right) \\
 &  + \frac{1}{k(i)2^{i+2}+1} \ln \left( k(i) 2^{i+1}+ 3/2 \right) 
\end{split}
\end{equation}

Therefore, if we call $F: \mathbb{R}^+ \to \mathbb{R}$ the function defined as:

\begin{equation}
F(k(i))= \frac{1}{ k(i)2^{i+2}-1 } \ln \left( k(i) 2^{i}+ 3/4 \right) + \frac{1}{k(i)2^{i+2}+1} \ln \left( k(i) 2^{i+1}+ 3/2 \right) - \tau(i)
\end{equation}

our problem reduces to determinate if $F$ has a computable zero. 

A straightforward calculation shows that for all $i \in \mathbb{N}_0$ we have $F(2^{-(i+1)})= \ln (625/32) - \tau(i)$, which, since by definition $2^{-i/2}\leq \tau(i) \leq 1$ implies $F(2^{-(i+1)}) >0$. 

On the other hand, for all $i \in \mathbb{N}_0$ we have $F(1) =  \frac{1}{ 2^{i+2}-1 } \ln \left( 2^{i}+ 3/4 \right) + \frac{1}{2^{i+2}+1} \ln \left( 2^{i+1}+ 3/2 \right) - \tau(i)$. Since for all $i \in \mathbb{N}_0$ we know that $$ \frac{1}{ 2^{i+2}-1 } \ln \left( 2^{i}+ 3/4 \right) + \frac{1}{2^{i+2}+1} \ln \left( 2^{i+1}+ 3/2 \right) < 2^{-i/2} $$ and that by definition $\tau(i) \geq 2^{-i/2}$, it follows that $F(1)< 0$. 

Because function $F$ is decreasing and computable in every interval $[2^{-(i+1)}, 1]$ for all $i \in \mathbb{N}_0$, it follows that any such interval contains one unique zero for such function, and thanks to \cite{Wei00}[Theorem 6.3.8] we know that such zero is computable. Therefore we can always pick such zero to play the role of coefficient $k(i)$ in Equation \ref{eq:haltderivative}, ensuring that if we have $y_1(t)=- 2^{-i}$ for any $t \in [0,5-\tau(i)]$ then $y_1(t+\tau(i))=- 2^{-(i+1)}$. 

Note that $f_1$ as defined above is then computable and bounded in each interval (i.e.\ when $x_1 \in [- 2^{-i},- 2^{-(i+1)}]$ for all $i \in \mathbb{N}_0$) since each $k(i)$ is. Moreover, this consideration extends trivially to the whole domain $E$ exept for the line $\{ (0,x_2) : x_2 \in [-5,5]\}$. The continuity of $f_1$ on such line depends on the behavior of $k(i)$ as $i$ tends to infinity. Precisely, $f_1$ is continuous on such line if and only if $\lim_{i \to \infty} k(i)=0$. 

Function $f_2$ is constructed in a similar manner; once again its action depends only on the value of its first component. First we require $f_2(x_1, x_2)=0$ if $x_1 \in [-5,-1]$ and  if $x_1 \in [0,5]$. Then we construct it on intervals of the form $[- 2^{-i},- 2^{-(i+1)}]$ for all $i \in \mathbb{N}_0$ making use of function $f_1$ as previously defined. More precisely, if $x_1 \in [- 2^{-i},- 2^{-(i+1)}]$ for some $i \in \mathbb{N}_0$ we define $f_2(x_1,x_2)= 2^{-h(i+1)-1} \pi f_1(x_1,x_2) \sin{\left( \pi \left(\frac{x_1 + 2^{-i}}{-2^{-(i+1)}+ 2^{-i}}\right) \right)}$. In this way we are making sure that variable $y_2$ increases of $2^{-h(i+1)}$ within the same time step $\tau(i)$. Indeed we have: 
\begin{equation}
\begin{split}
 \int_{0}^{\tau(i)} f_2(x_1,x_2) dt  & = 2^{-h(i+1)-1}  \pi \int_{0}^{\tau(i)}  x_1^{'} \sin{\left( \pi \left(\frac{x_1 + 2^{-i}}{-2^{-(i+1)}+ 2^{-i}}\right)\right)} dt \\
 \Rightarrow y_2 (\tau(i)) - y_2 (0)  & =  2^{-h(i+1)-1} \pi \int_{- 2^{-i}}^{- 2^{-(i+1)}} \sin{\left( \pi\left( \frac{x_1 + 2^{-i}}{-2^{-(i+1)}+ 2^{-i}} \right) \right)} dx_1 \\
 & =  2^{-h(i+1)} 
\end{split}
\end{equation}

It is easy to see that $f_2 (-2^{-i},x_2) = 0$ for all $i \in \mathbb{N}_0$, for all $x_2 \in [-5,5]$. Therefore, note that $f_2$ as defined above is computable and bounded everywhere exept on the line $\{ (0,x_2) : x_2 \in [-5,5]\}$, where its continuity properties follows from the ones of $f_1$ on such line. 

We now need to verify that, for both $y_1$ and $y_2$, the derivatives exist at time $\tau^*$. In other words, we need to verify that, if the components of the solution $y$ behave accordingly to this dynamic, then $y$ is differentiable at time $\tau^*$. We start with variable $y_1$. By definition the derivative in $\tau^*$ is defined as $y'_1(\tau^*)= \lim_{t \to \tau^*} D_{y_1}(t) = \frac{y_1 (\tau^*) - y_1 (t)}{\tau^* - t}$. For all $k \in \mathbb{N}_0$ consider the quantity: $D_{y_1}(\tau(k)) = \frac{y_1 (\tau^*) - y_1 (\tau(k))}{\tau^* - \tau(k)}= \frac{2^{-(k+1)}}{\sum_{i=k+1}^{\infty} \tau(i)}$. We want to show that $ D_{y_1}(\tau(k))$ tends to zero as $k$ tends to infinity. This is enough to prove the existence of the derivative at time $\tau^*$ since $y_1$ as described is increasing. To do so we focus on bounding the quantity $\frac{2^{-(k+1)}}{\tau(k+1)}$ since clearly $ D_{y_1}(\tau(k)) < \frac{2^{-(k+1)}}{\tau(k+1)}$. By construction of function $\tau$, we have two different cases, one in which $h(k+1) < k+1$ and one in which $h(k+1) \geq k+1$. If $h(k+1) < k+1$ then $\tau(k+1)= 2^{-h(k+1)/2}$ and hence $\frac{2^{-(k+1)}}{\tau(k+1)} = 2^{-(k+1) + h(k+1)/2}$ from which it follows that $\frac{2^{-(k+1)}}{\tau(k+1)} < 2^{-(k+1)/2}$. Instead, if $h(k+1) \geq k+1$ then $\tau(k+1)= 2^{(k+1)/2}$ and hence $\frac{2^{-(k+1)}}{\tau(k+1)} = 2^{-3(k+1)/2 }$. Therefore, in both cases we have $\frac{2^{-(k+1)}}{\tau(k+1)} < 2^{-(k+1)/2}$ which implies $D_{y_1}(\tau(k)) < 2^{-(k+1)/2}$. Consequently, $y'_1(\tau^*)= \lim_{t \to \tau^*} D_{y_1}(t)= \lim_{k \to \infty} D_{y_1}(\tau(k)) =0$. 

We proceed with variable $y_2$. By definition the derivative in $\tau^*$ is defined as $y'_2(\tau^*)= \lim_{t \to \tau^*} D_{y_2}(t) = \frac{y_2 (\tau^*) - y_2 (t)}{\tau^* - t}$. For all $k \in \mathbb{N}_0$ consider the quantity: $D_{y_2}(\tau(k)) = \frac{y_2 (\tau^*) - y_2 (\tau(k))}{\tau^* - \tau(k)}= \frac{\sum_{i=k+2}^{\infty} 2^{-h(i)}}{\sum_{i=k+1}^{\infty} \tau(i) }$. We want to show that $ D_{y_2}(\tau(k))$ tends to zero as $k$ tends to infinity. Once again, this is enough to prove the existence of the derivative at time $\tau^*$ since $y_2$ as described is increasing. To do so we focus on bounding the quantity $\frac{2^{-h(i)}}{\tau(i)}$ for all $i>k$ since clearly $ D_{y_2}(\tau(k)) < \sum_{i=k+1}^{\infty} \frac{2^{-h(i)}}{\tau(i)}$. By definition of function $\tau$, for all $i>k$ we have two different cases, one in which $h(i) < i$ and one in which $h(i) \geq i$. If $h(i) < i$ then $\tau(i)= 2^{-h(i)/2}$ and hence $\frac{2^{-h(i)}}{\tau(i)} = 2^{-h(i)/2}$. Instead, if $h(i) \geq i$ then $\tau(i)= 2^{i/2}$ and hence $\frac{2^{-h(i)}}{\tau(i)} \leq 2^{-3i/2}$. Therefore, in both cases we have $\frac{2^{-h(i)}}{\tau(i)} \leq 2^{-3i/2} + 2^{-h(i)/2}$ which implies $ D_{y_2}(\tau(k)) \leq \sum_{i=k+1}^{\infty}  2^{-3i/2} + 2^{-h(i)/2}$. But since $\sum_{i=k+1}^{\infty}  2^{-3i/2} + 2^{-h(i)/2}$ is the remainder of a converging serie, we have that $\lim_{k \to \infty} \sum_{i=k+1}^{\infty}  2^{-3i/2} + 2^{-h(i)/2} = 0$. Consequently, $y'_2(\tau^*)= \lim_{t \to \tau^*} D_{y_2}(t)= \lim_{k \to \infty} D_{y_2}(\tau(k)) =0$. 

We have just shown that the derivatives of variables $y_1$ and $y_2$ at time $\tau^* $ both exist and equal to zero. Moreover, by continuity we have $y_1(\tau^* )=0$ and $y_2(\tau^* )=\mu$.

By construction the IVP considered in Equation \ref{eq:alphaIVP} has a unique solution $y:[0,5] \to E$ and since $f_1(x_1, x_2)=f_2(x_1, x_2)=0$ when $x_1 \in [0,5]$, we have $y_2(5)=\mu$ which is the noncomputable real number encoding the halting problem according to Definition \ref{def:noncompu}. Since the solution of a computable IVP must be computable according to \cite{collins2009effective}, it follows that in Equation \ref{eq:haltderivative} we must have $\lim_{i \to \infty} k(i) \neq 0$, which implies that the set of discontinuity points of $f$ is the straight-line $\{ (0,x_2) : x_2 \in [-5,5]\} \subset E$. Finally, it is clear that $f$ is a solvable function, since it is a function of class Baire one and it is identically equal to zero on the mentioned discontinuity line. 

\end{proof}

The above theorem sets the premises to construct more complex simulations involving solvable IVPs which yield as output any hyperarithmetical real. A real number $x$ is called hyperarithmetical if its left cut, i.e.\ the set of rationals $\{ q \in \mathbb{Q} : q < x \}$ is hyperarithmetical. In analogy with the techniques discussed at the end of Section \ref{sec:examples} for building complex examples of solvable systems, it is clear that the simulation described above for the halting set represents the key building block to be exploited transfinitely (by means of proper rescaling and concatenation) with the goal of yieding hyperarithmetical reals in any higher level of the hierarchy. More on this topic is discussed in the conclusions below. 

\section{Conclusion}
\label{sec:conclusion}

We have explored the characteristics of initial value problems (IVPs) involving discontinuous ordinary differential equations (ODEs) that have a unique solution. Through our research, we've delineated a particular class of these systems, which we have termed \emph{solvable}. For such systems, the solution can always be obtained analytically through a transfinite recursion process, with a countable number of steps. 

We have presented various instances of these systems and described a methodology for constructing increasingly complex examples. We have shown that even basic configurations of these solvable systems, such as a simple set of discontinuity points, can produce noncomputable values and address the halting problem.

Our approach bears similarities to Denjoy's method for tackling the challenge of antidifferentiation, and in consideration of the results and methodologies discussed in studies like \cite{dougherty1991complexity} and \cite{Wes20}, we believe we have laid the groundwork for a more detailed analysis of these systems. This could potentially lead to establishing a classification and hierarchy within this robust class of solvable systems.

%

More precisely: 
the integrability rank that inspired the modus operandi of our ranking is the Denjoy rank. Now: \begin{itemize}
\item An analysis has been done for the Denjoy rank: an unpublished theorem from Ajtai, whose proof is included in \cite{dougherty1991complexity}, demonstrates that once a code for a derivative $f$ is given as a computable sequence of computable functions converging pointwise to $f$, then the antiderivative $F$ of $f$ is $\Pi_1^1$ relative to $f$. This fact has direct implications related to the hierarchy of hyperarithmetical reals (also called $\Delta_1^1$ reals) expressed formally by a Theorem from \cite{dougherty1991complexity}, which proves that the hyperarithmetical reals are exactly those reals $x$ such that $x=\int_{0}^{1} f$ for some derivative $f$ of which we know the code of. 
\item An alternative computability theoretic analysis  of   Denjoy's rank has been done in \cite{Wes20}, relating levels of the hierarchy to levels of the arithmetical hierarchy in some precise manner, using a slightly alternative setting, on the way objects are encoded. 
\end{itemize}

We posit that applying a similar analytical approach to the framework of solvable systems could lead to refined statements regarding the ranking and ordinals involved, tailored to specific classes of functions or dynamics.

This analytical perspective, in conjunction with research detailed in papers like \cite{TAMC06} and \cite{JournalACM2017}, which explore the simulation of discrete computational models through analog models based on ODE systems, suggests the potential of solvable IVPs as an analog framework for simulating transfinite computations or offering an alternative method for presenting such computations.

Moreover, our discussions have highlighted categories of ordinary differential equations that exhibit complexities not commonly covered in traditional texts on ODEs, along with numerous existing counterexamples in the literature. In line with arguments from \cite{dougherty1991complexity}, our findings indicate that the complete set of countable ordinals is essential in any constructive process for solving ODEs in general. Furthermore, for each countable ordinal, we can construct a solvable ODE with corresponding complexity. This underscores the depth and novelty our research contributes to the field of differential equations.
\bibliographystyle{plainurl}
\bibliography{bournez,perso}
\end{document}